\documentclass[11pt]{article}
\usepackage{amsmath}
\usepackage{amssymb}
\usepackage{amsthm}
\usepackage{ifthen}
\usepackage{xcolor}
\usepackage{dsfont}
\usepackage{ccfonts,eulervm} \usepackage[T1]{fontenc}\usepackage{hyperref}
\hypersetup{colorlinks=true,linkcolor=black,citecolor=[rgb]{0.5,0,0}}
\usepackage[margin=1in]{geometry}
\usepackage{bm}
\usepackage{nicefrac}
\usepackage{footmisc}

\newcommand{\F}{\mathbb{F}}
\newcommand{\nc}{\newcommand}
\nc{\rnc}{\renewcommand} \nc{\nev}{\newenvironment}

\renewcommand{\ge}{ \geqslant}
\renewcommand{\le}{ \leqslant}

\nc{\W}[1]{\text{$\textup{W}[#1]$}}
\nc{\FPT}{\textup{FPT}}
\nc{\fpt}{\textup{fpt}}

\nc{\dotcup}{\;\dot\cup\;}

\newtheorem{theorem}{Theorem}[section]
\newtheorem{lemma}[theorem]{Lemma}

\newtheorem{proposition}[theorem]{Proposition}
\newtheorem{fact}[theorem]{Fact}

\newtheorem{definition}[theorem]{Definition}

\newtheorem{claim}[theorem]{Claim}

\makeatletter
\newcommand\footnoteref[1]{\protected@xdef\@thefnmark{\ref{#1}}\@footnotemark}
\makeatother

\newcommand{\OPT}{\textsf{OPT}}
\rnc{\c}{\mathbf{{c}}} 
\rnc{\a}{\mathbf{{a}}}
\rnc{\b}{\mathbf{{b}}} 
\rnc{\i}{\mathbf{{i}}}
\rnc{\j}{\mathbf{{j}}} 
\rnc{\u}{\mathbf{{u}}}
\rnc{\v}{\mathbf{{v}}}
\nc{\B}{\mathbf{{B}}}
\nc{\s}{\mathbf{{s}}}
\rnc{\S}{\mathbf{{S}}}
\nc{\C}{\mathbf{{C}}}
\nc{\U}{\mathbf{{U}}}
\nc{\E}{\mathbf{{E}}}
\nc{\G}{\mathbf{{G}}}
\nc{\D}{\mathbf{{D}}}
\nc{\unsat}{\mathbf{{UNSAT}}}
\renewcommand{\epsilon}{\varepsilon}
\newcommand{\poly}{\text{poly}}
\newcommand{\alphab}{\vec{\bm{\alpha}}}
\newcommand{\betab}{\vec{\bm{\beta}}}

\newcommand{\var}{\mathsf{var}}
\renewcommand{\tilde}{\widetilde}
\renewcommand{\kappa}{\tau}
\newcommand{\rhob}{\vec{\bm{\rho}}}

\newcommand{\M}{\mathcal{M}}

\usepackage{setspace}
\author{Karthik C.\ S.\footnote{This work was supported by a grant from the Simons Foundation, Grant Number 825876, Awardee Thu D. Nguyen. Also, part of this work was done when the author was a postdoctoral researcher at New York University and supported by Subhash Khot's Simons Investigator Award.}\\ Department of Computer Science\\ Rutgers University\\\texttt{karthik.cs@rutgers.edu} \and Subhash Khot\footnote{This work was supported by the NSF Award CCF-1422159,
2130816, the Simons Collaboration on Algorithms and Geometry, and the
Simons Investigator Award.
}\\ Department of Computer Science\\ Courant Institute of Mathematical
Sciences\\ New York University\\ \texttt{khot@cims.nyu.edu} }
\title{Almost Polynomial Factor Inapproximability for Parameterized $k$-Clique}
\date{}

\begin{document}

\maketitle

\begin{abstract} 
The $k$-Clique problem is a canonical hard problem in parameterized complexity. 
In this paper, we study the parameterized complexity of approximating the $k$-Clique  problem where an integer $k$ and a graph $G$ on $n$ vertices are given as input, and the goal is to find a clique of size at least $k/F(k)$ whenever the graph $G$ has a clique of size $k$. When such an algorithm runs in time $T(k) \cdot \text{poly}(n)$ (i.e., FPT-time) for some computable function $T$, it is said to be an \emph{$F(k)$-FPT-approximation algorithm} for the $k$-Clique problem. \vspace{0.1cm}

Although  the non-existence of an $F(k)$-FPT-approximation algorithm for any computable  sublinear function $F$ is known under gap-ETH [Chalermsook~et~al., FOCS 2017], it has remained a long standing open problem to prove the same inapproximability result under the more standard and weaker assumption, W[1]$\neq$FPT.\vspace{0.1cm}
 
In a recent breakthrough,  Lin [STOC 2021] ruled out  constant factor (i.e., $F(k)=O(1)$) FPT-approximation algorithms  under W[1]$\neq$FPT. In this paper, we improve this inapproximability result (under the same assumption) to rule out  every $F(k)=k^{1/H(k)}$ factor FPT-approximation algorithm for any increasing computable function $H$ (for example $H(k)=\log^\ast k$).   \vspace{0.1cm}

Our main technical contribution is introducing list decoding of Hadamard codes over large prime fields  into the proof framework of Lin. 
\end{abstract}
\clearpage

\section{Introduction}
In the {\em clique} problem (Clique), we are given an undirected graph $G$ on $n$ vertices and an integer $k$, and the goal is to decide whether there is a subset of vertices $S\subseteq V(G)$ of size $k$ such that every two distinct vertices in $S$ share an edge in $G$. 
Often regarded as one of the classical problems in computational complexity, Clique was first shown to be NP-complete in the seminal work of Karp~\cite{Karp72}. Thus, its optimization variant, namely the {\em maximum clique}, where the goal is to find a clique of the largest possible size, is also  NP-hard. 

To circumvent this apparent intractability of the problem, the study of an approximate version was initiated. The quality of an approximation algorithm is measured by the \emph{approximation ratio}, which is the ratio between  the size of the maximum clique and the size of the solution output by the algorithm. It is trivial to obtain an $n/c$ factor approximation algorithm for any constant $c\in\mathbb{N}$. The state-of-the-art approximation algorithm is due to Feige \cite{F04} which yields an approximation ratio of $O(n(\log\log n)^2/\log^3 n)$. 
On the opposite side, Maximum Clique is arguably the first natural combinatorial optimization problem studied in the context of hardness of approximation; in a seminal work of Feige, Goldwasser, Lov\'asz, Safra and Szegedy \cite{FGLSS96}, a connection (hereafter referred to as the FGLSS reduction) was made between interactive proofs and hardness of approximating Clique. The FGLSS reduction, together with the PCP theorem \cite{AS98,ALMSS98,D07} and gap amplification via randomized graph products \cite{BS92}, immediately implies $n^\varepsilon$ ratio inapproximability of Clique for some constant $\varepsilon>0$ under the assumption that NP$\not\subseteq$ BPP. Following \cite{FGLSS96},   a long line of research on the inapproximability of Clique \cite{BGLR93,FK00,BGS98,BS94},   culminated in the works of H\r{a}stad  \cite{H96a,H96b}, wherein it was shown that Clique cannot be approximated to within a factor of $n^{1-\varepsilon}$ in polynomial time unless NP$\subseteq$ZPP; this was later derandomized by Zuckerman  \cite{Z07}. Since then, better inapproximability ratios are known \cite{EH00,K01,KP06}, with the best ratio being $n/2^{(\log n)^{3/4+\varepsilon}}$ for every $\varepsilon> 0$ (assuming NP $\nsubseteq$ BPTIME$(2^{(\log n)^{O(1)}})$) due to Khot and Ponnuswami \cite{KP06}. In summary, our understanding of the boundaries of efficient computation of approximating clique in the NP world is almost complete.

Besides approximation, another widely-used technique to cope with NP-hardness is \emph{parameterization}. The parameterized version of Clique, which we will refer to simply as $k$-Clique, is exactly the same as the original decision version of the problem except that now we are not looking for a polynomial time algorithm but rather a {\em fixed parameter tractable} (FPT) algorithm -- one that runs in time $T(k)\cdot \poly(n)$ for some computable function $T$ (e.g., $T(k) = 2^k$ or even $2^{2^k}$). Such running time will henceforth be referred to as \emph{FPT time}. It turns out that even with this relaxed requirement, $k$-Clique still remains intractable: in the same work that introduced the $W$-hierarchy, Downey and Fellows~\cite{DF95} showed that $k$-Clique is complete for the class W[1], which is generally believed to not be contained in FPT, the class of fixed parameter tractable problems. Subsequently, stronger running time lower bounds have been shown for $k$-Clique under stronger assumptions. Specifically, Chen et al.~\cite{CHKX06a,CHKX06b} ruled out $T(k) \cdot n^{o(k)}$-time algorithms for $k$-Clique assuming the {\em Exponential Time Hypothesis} (ETH)\footnote{ETH~\cite{IP01,IPZ01} states that no subexponential time algorithm can solve 3-SAT.}.  Note that the trivial algorithm that enumerates through every $k$-tuple of vertices, and checks whether it forms a clique, runs in $\tilde{O}(n^{k})$ time. It is possible to speed up this running time using fast matrix multiplication~\cite{NP85,EG04}.

Given the strong negative results for $k$-Clique discussed in the previous paragraph, it is natural to ask whether one can come up with a \emph{fixed parameter approximation (FPT-approximation) algorithm} for $k$-Clique. The notion of FPT-approximation algorithms is motivated primarily through the  consideration of inputs with small sized optimal solutions. Case in point, the state-of-the-art polynomial time approximation ratio of $O(n(\log\log n)^2/\log^3 n)$ \cite{F04} would be meaningless if the size of the maximum clique (denoted $\OPT$) was itself  $O(n(\log\log n)^2/\log^3 n)$, as outputting a single vertex already guarantees an $\OPT$-approximation ratio. In this case, a bound such as $o(\OPT)$ would be more meaningful. Unfortunately, no approximation ratio of the form $o(\OPT)$ is known even when FPT-time is allowed.  We refer the reader to the textbooks~\cite{DowneyF13,CyganFKLMPPS15} for an excellent introduction to the area.
On the other hand, inapproximability results in parameterized complexity aim to typically rule out algorithms running in FPT time (under the W[1]$\neq$FPT\ hypothesis) for various classes of computable functions $F$. 
This brings us to the main question addressed in our work:
\begin{center}
\emph{Is there an $F(k)$-FPT-approximation algorithm for $k$-Clique\\ for some computable function $F$ which is $o(k)$?}
\end{center}
This question, which dates back to late 1990s (see, e.g., remarks in \cite{DFM06}), has attracted significant attention in literature and continues to be repeatedly raised in workshops and surveys on parameterized complexity \cite{CGG06,marx2008fpa,CHK13,BEOP15,saket2017,CL19,KLM19,FKLM20}. This open problem is  
even listed\footnote{In \cite{DowneyF13}, the authors list proving hardness of approximation for dominating set as  one of the six ``most infamous'' open questions in the area of Parameterized Complexity. Immediately, they clairfy that, ``One can ask similarly about an FPT Approximation for Independent Set''. Note that inapproximability results for  independent set problem imply  hardness of approximation of $k$-Clique and vice versa. } in the seminal textbook of Downey and Fellows~\cite{DowneyF13}.

Early attempts \cite{HKK13,BEOP15} ruled out constant ratio FPT-approximation algorithms for $k$-Clique, but under very strong assumptions such as the combination of ETH and the existence of a linear-size PCP. However, a few years ago, the authors in \cite{CCKLMNT20} proved under the {\em Gap Exponential Time Hypothesis} (Gap-ETH)\footnote{Gap-ETH~\cite{D16,MR16} is a strengthening of ETH, and states that no subexponential time algorithm can distinguish satisfiable 3-CNF formulae from ones that are not even $(1 - \varepsilon)$-satisfiable for some $\varepsilon > 0$.}, that no $F(k)$-approximation algorithm for $k$-Clique exists for any computable function $F$. %, thereby negatively answered the question. 
Such non-existence of FPT-approximation algorithms is referred to in literature as the \emph{total FPT-inapproximability} of $k$-Clique.

While the result in \cite{CCKLMNT20}  seems to settle the parameterized complexity of approximating $k$-Clique, there are a few disadvantages to their result. First, while Gap-ETH may be plausible, it is a strong conjecture and in their reduction, the hypothesis does much of the work in the proof. In particular, Gap-ETH itself already gives the gap in hardness of approximation; once they have  such a gap, it suffices to design gap preserving reductions to prove other inapproximability results (although some care needs to be taken as they cannot directly use Raz's parallel repetition theorem \cite{R98} for gap amplification). This is analogous to the NP-world, where once one inapproximability result can be shown, many others follow via relatively simple gap-preserving reductions (see, e.g., \cite{PY91}). However, creating a gap in the first place requires the PCP Theorem~\cite{AS98,ALMSS98,D07}, which involves several new technical ideas such as local checkability and decodability of codes and proof composition. Hence, it is desirable to bypass Gap-ETH and prove total FPT-inapproximability under a standard assumption such as W[1]$\neq$FPT, that does not inherently have a gap.

The last seven years have witnessed many significant inapproximability results in parameterized complexity that are only based on the assumption  W[1]$\neq$FPT. A key component in all these works is a gap creating technique. Elaborating, we now have strong inapproximability results under W[1]$\neq$FPT for 
Set Cover \cite{CL19,KLM19,L19,KN21}, Set Intersection \cite{L18,BKN21},
Steiner Orientation problem \cite{W20}, and problems in Coding theory and Lattice theory \cite{BBEGKLMM21}. There have been even more   strong inapproximability results under Gap-ETH proved in these last few years and we direct the reader to a recent survey \cite{FKLM20} on the topic.

Returning to the question of proving hardness of approxmation of $k$-Clique, the difficulty in adopting the  techniques developed in the NP world into parameterized complexity 	were discussed in many previous works, such as \cite{CGG06,L18,CL19, FKLM20}, and it was also widely believed \cite{CGG06} that one needs to prove a PCP theorem analogue for parameterized complexity\footnote{One such formulation is called the \emph{Parameterized Inapproximability Hypothesis} (PIH) and was putforth by \cite{LRSZ20}. See Section~\ref{sec:open} for a small discussion on PIH. } in order to obtain any non-trivial inapproximability result for $k$-clique under the W[1]$\neq$FPT assumption. Recently, in a remarkable breakthrough, Lin \cite{L21} negated this belief, and developed a different proof framework to prove constant ratio inapproximability for the $k$-Clique problem assuming W[1]$\neq$FPT. 

We describe Lin's proof framework in detail in Section~\ref{sec:overview}, but even given the result of \cite{L21}, it is still very far from proving the total FPT-inapproximability of $k$-Clique. Our result below is a significant improvement on the state-of-the-art W[1]-hardness of approximation of $k$-Clique.

\begin{theorem}[Almost Polynomial Factor Inapproximability of $k$-Clique]\label{thm:main}
Let $H:\mathbb{N}\to\mathbb{N}$ be an increasing\footnote{\label{footnote:increasing}A function $H:\mathbb{N}\to\mathbb{N}$ is said to be increasing if for all $k\in\mathbb{N}$ we have $H(k+1)\ge H(k)$, and $\underset{k\to\infty}{\lim}\ H(k)=\infty$.} computable function such that $\forall k\in\mathbb{N}$, we have $H(k)\le \log k$. Given as input an integer $k$ and a graph $G$ on $n$ vertices, it is W[1]-hard parameterized by $k$ (under randomized reductions), to distinguish between the following two cases:
 \begin{description}
 \item[Completeness:]$G$ has a clique of size $k$.
 \item[Soundness:] $G$ does not have a clique of size $\nicefrac{k}{k^{1/H(k)}}$.
 \end{description}
\end{theorem}

For example, if we plug in $H(k)=\lceil \log\log k\rceil $  in our theorem, we obtain $k^{1/\lceil \log\log k\rceil}=\omega(\text{polylog } k)$ ratio inapproximability of $k$-Clique. In fact, if we susbstitute $H$ in the theorem statement with a \emph{very} slowly growing function, then we \emph{almost} obtain polynomial ratio inapproximability of $k$-Clique. We reiterate again that the only comparable result to the above theorem, is by Lin  \cite{L21}, who ruled out constant ratio (i.e., $H(k)=O(\log k)$) FPT-approximation algorithms.  

Our result also rules out $k^{1/H(k)}$ ratio FPT-approximation algorithms for the $k$-Independent Set problem by using the well-known connection to the $k$-Clique problem. 

We remark here that independent of our work, in \cite{LRSW21}, the authors assuming ETH, rule out FPT algorithms for approximating the $k$-clique problem to the  same hardness of approximation factors as in Theorem~\ref{thm:main}, and this has been subsequently improved in \cite{LinRSW23,GuruswamiLRS024}. Note that $W[1]\neq \FPT$ is a weaker assumption than ETH as the latter is known to imply the former \cite{CHKX06a,CHKX06b}. 
Finally, subsequent to our work, in \cite{CFLL25}, the authors have provided an alternate and significantly simpler proof of Theorem~\ref{thm:main}.

 \subsection{Proof Overview}\label{sec:overview}
 In this subsection, we provide a proof overview of Theorem~\ref{thm:main}. In order to motivate our proof framework and ideas, we first describe a wishful thinking reduction to gap $k$-Clique, and then describe Lin's framework, and finally provide the details of our techniques.  
 
 \paragraph{From PIH to gap $k$-Clique.} Suppose our starting point was a gap 2-CSP\footnote{A $t$-CSP is a constraint satisfaction problem in which every constraint involves at most $t$ variables.} instance $\varphi$ on $k$ variables and alphabet $[n]$, which is either completely satisfiable (i.e., there exists an assignment that satisfies all the constraints) or every assignment to the variables violates at least 1\% of the constraints.  Furthermore, suppose that  it was W[1]-hard, parameterized by $k$, to decide $\varphi$. This assumption is known as PIH and it was believed \cite{CGG06} that we need to first prove PIH in order to prove the hardness of gap $k$-Clique. Applying the well-known FGLSS reduction to $\varphi$, we 
 obtain a graph in which finding a clique which is larger than 99\% of the maximum clique size is W[1]-hard. Of course, the big problem with this reduction is that we do not know if PIH is true.

\paragraph{Lin's Framework.} In \cite{L21}, the author circumvents proving PIH, and instead makes the following surprising observation. Let $\varphi$ be a 2-CSP  instance where the variable set is thought of as $\{0,1\}^k$, and the constraints are only between a pair of points that differ  on one coordinate. We call a constraint to be in direction $i\in[k]$ if the constraint is between a pair of points that differ on the $i^{\text{th}}$ coordinate. Suppose we can show that it is W[1]-hard parameterized by $t:=2^k$, to distinguish between the cases when either $\varphi$ is satisfiable or when, for every assignment to $\{0,1\}^k$, there exists $i\in[k]$, such that 1\% of the constraints in the $i^{\text{th}}$ direction are violated. Note that in the soundness case, there is no guarantee that for every  assignment, 1\% of the total constraints are violated, in fact, for every  assignment we are only guaranteed that  $\Omega(1/k)$ fraction of the total constraints are violated. 
Nevertheless, by applying a simplified FGLSS reduction\footnote{The FGLSS reduction that we refer here is a simplified reduction specifically for 2-CSPs. Starting from a 2-CSP instance $\varphi$ on $k$ variables and alphabet set $[n]$, we construct the graph $G_{\varphi}$ whose vertex set is $[k]\times [n]$ where we do not have an edge between $(i,a)$ and $(j,b)$ if and only if the constraint between $i$ and $j$ in $\varphi$ is not satisfied by $(a,b)$.} to $\varphi$, we obtain a gap $t$-Clique! 

Therefore, informally speaking, a wishful version of Lin's framework comprises of two steps.

(i)  Show W[1]-hardness of deciding 2-CSP on the Boolean hypercube host graph with the aforementioned soundness property. 

(ii) Apply the FGLSS reduction to reduce the above 2-CSP to the gap $t$-Clique problem because, any subset of the Boolean hypercube containing 99.5\% of the vertices, also contains at least 99\% of the edges in every direction.

Lin starts from the $k$-Vector Sum problem, where given $k$ collections of $n$ Boolean vectors each, the goal is to decide if there are $k$ vectors, one in each collection, that sum to $\vec{0}$. 
Starting from  the $k$-Vector Sum problem and by using the local testability and local decodability of Hadamard codes over $\F_2$, he shows the  W[1]-hardness of deciding 3-CSP on some variant of the Boolean hypercube host graph, with the aforementioned soundness property. 

However, since we have a CSP of arity three, applying the simplified FGLSS reduction for 2-CSPs directly becomes tricky,  and he finds a critical modification to the FGLSS reduction, which allows him to reduce to the gap $t$-Clique problem. We note that the gap created is between the existence of  a $t$-clique in the completeness case versus no $0.995t$-clique in the soundness case. In order to rule out FPT-approximation algorithms for all constant ratios, he applies the well-known technique of   graph product, by taking an $O(1)$-wise product  of the hard $k$-Clique instance and the size of the graph increases only to $n^{O(1)}$.

\paragraph{Our Framework.} We are now ready to describe our proof framework. At a high level, the gap created by Lin mainly arrives from the distance of the Hadamard code. Since the gap generated by using Hadamard codes over $\F_2$ is at most $1/2$,  in order to obtain larger gaps, we use Hadamard codes over $\F_q$, for some large  $q$ only depending on $k$.  However, working with Hadamard codes over $\F_q$ in the low acceptance regime, has its own challenges, such as:

$\bullet$ First, in Lin's case, local testability of Hadamard codes in the high acceptance regime is just the standard BLR Linearity testing \cite{BLR93}, which can be used off the shelf. However, we need to test the Hadamard code in the low acceptance regime over $\F_q$, and thus we prove results on the list decodability of Hadamard codes over $\F_q$. Such results appear implicitly in literature, and we make them explicit through our Theorems~\ref{thm:lintest}~and~\ref{thm:piecing}. 

$\bullet$ Second, because we deal with list decoding instead of standard decoding, all the relationships in our proofs have some ``noise'' and therefore the arguments in our soundness analysis of Theorem~\ref{thm:main} are very intricate. 

We have described above, the challenges that we had to address to prove Theorem~\ref{thm:main} over the result of \cite{L21}. Next, we sketch the outline of our proof.

Our starting point is the same as Lin's starting point, i.e., the $k$-Vector Sum problem, but over $\F_q$. The W[1]-hardness of the $k$-Vector Sum problem is known in literature \cite{ALW13}, and in fact Lin provides a short proof in his paper. Then we create a 3-CSP on the variable set $\F_q^k$ and alphabet size $[n]$ with three types of constraints:

(i) We have 3-ary constraints arising from the 3-query list decoding of Hadamard codes. These constraints enforce that the assignments satisfying them  can themselves be viewed as a Hadamard codeword.  In particular, for every $k$-tuple of vectors of the $k$-Vector Sum  instance, our assignment is supposed to be the Hadamard encoding of the sum of the $k$-tuple of vectors. 

(ii) We have 2-ary constraints arising from a pair of points on any  axis parallel line in $\F_q^k$. The constraints along the $i^{\text{th}}$ direction enforce  that the  $i^{\text{th}}$  vector in our $k$-tuple of vectors indeed comes from  the  $i^{\text{th}}$ collection in the $k$-Vector Sum  instance. 

(iii) We have 2-ary constraints arising from a pair of points on specific lines through the origin, which enforce that the  sum of the $k$-tuple of vectors 
is $\vec{0}$.

After constructing this CSP, we build an instance of the $t$-Clique problem, where $t:=q^{2k}$, by building a graph on $q^{2k}$ clouds of vertices, where each cloud is an independent set containing one vertex for each triple $(x,y,x+y)\in[n]\times [n]\times [n]$. Each cloud represents a pair of variables of our CSP, which are the queries to the linearity test. The satisfying pairs of the alphabet set of the constraints in items (ii) and (iii) appear directly as edges in the graph. Since every variable appear in multiple clouds of vertices, we only put an edge between pairs of vertices that are ``consisitent'' on their assigment to a variable. 

Unlike \cite{L21}, we do not analyze the reduction from $k$-Vector Sum problem to the 3-CSP and from the 3-CSP to the $t$-Clique problem in two separate steps, but rather we analyze the instance of the $t$-Clique directly with respect to the $k$-Vector Sum problem, and this helps us keep the analysis clean and succinct. 
A more detailed overview of this reduction and analysis  is given in Section~\ref{sec:construction}. 

\subsection{Organization of Paper}
The paper is organized as follows. 
First, in Section~\ref{sec:prelim} we define the $k$-Vector Sum problem and state its known hardness result.
Then, in Section~\ref{sec:lintest} we prove linearity testing result in the low soundness regime (a.k.a.\ list decoding of Hadamard code) over fields of large prime order.
Next, in Section~\ref{sec:main} we prove our main result, i.e., Theorem~\ref{thm:main}.
Finally, in Section~\ref{sec:open} we highlight a couple of important open problems. 

\section{Preliminaries}\label{sec:prelim}

First, we define the notion of relative Hamming norm that is used throughout this paper. 
Let $q$ be a prime power and $n\in\mathbb{N}$. For any vector $x\in \F_q^d$ we define its relative Hamming norm, denoted $\|x\|$, as the fraction of coordinates in $[d]$ in which $x$ is not 0, i.e., 
$$\|x\|:=\frac{|\{i\in[d]:x_i\neq 0\}|}{d}.
$$

Next, we define the $k$-Vector Sum problem and state its known W[1]-hardness result. 

\begin{definition}[$k$-Vector Sum]
Let $q$ be a prime.
Given $k$ sets $U_1,\ldots,U_k$ of vectors  in $\mathbb{F}_q^{m}$, the goal of $k$-vector-sum problem is to decide whether there exist $\vec{u}_1\in U_1,\ldots,\vec{u}_k\in U_k$ such that
\[
\sum_{i\in[k]}\vec{u}_i=\vec{0}.
\]
\end{definition}

It is known that the above problem is W[1]-hard over finite fields \cite{ALW13}. We direct the reader to \cite{L21} for a short proof\footnote{The proof in \cite{L21} is over $\F_2$ but it is easy to see that their reduction generalizes to fields of larger characteristic. Also, they prove the hardness result for a version of $k$-Vector Sum where  a target vector is given as input, but that version reduces to the version given in this paper by simply including an extra collection containing only the negative of the target vector.  }.

\begin{theorem}\cite{ALW13,L21}\label{thm:kvec}
For every prime $q$ (independent of $n$), 
$k$-Vector-Sum over $\mathbb{F}_q$  and $m=\Theta(k^2\log n)$ is $W[1]$-hard parameterized by $k$.  
\end{theorem}

\section{Low Soundness Linearity Testing over Large Characteristic Fields}\label{sec:lintest}

In this section, we prove a linearity testing result which is a key  technical component in proving our inapproximability result. 

Let $q$ be a prime number and $d,\ell\in\mathbb{N}$.
Given a function $f:\F_q^d\to\F_q^{\ell}$, consider the following test $\mathcal{T}$. 
Pick $\alphab,\betab\in \F_q^d$ uniformly and independently at random. Accept if $f(\alphab)+f(\betab)=f(\alphab+\betab)$ and reject otherwise. Further, we define $S_{f,\mathcal{T}}\subseteq \F_q^d\times \F_q^d$ and $\var(f,\mathcal{T})\subseteq \F_q^d$ as follows:
$$
S_{f,\mathcal{T}}:=\{(\alphab,\betab)\in\F_q^d\times \F_q^d: f(\alphab)+f(\betab)=f(\alphab+\betab) \},
$$
$$
\var(f,\mathcal{T}):=\{\alphab\in\F_q^d: \exists \betab \in \F_q^d \text{ such that }(\alphab,\betab)\in S_{f,\mathcal{T}} \}.
$$
We say that a function $c:\F_q^d\to\F_q^{\ell}$ is linear if for all $\alphab,\betab\in \F_q^d$ we have $c(\alphab)+c(\betab)=c(\alphab+\betab)$. 
Moreover, we say that a function $f:\F_q^d\to\F_q^{\ell}$ is scalar respecting if for all $\alphab\in \F_q^d$ and all $\gamma\in\F_q$ we have $f(\gamma\cdot \alphab)=\gamma\cdot f(\alphab)$.

We prove below a couple of theorems in the flavor of the many list-decoding  results known in literature for Hadamard codes \cite{GL89,GRS00,G09}. 
 
\begin{theorem}[Linearity Testing]\label{thm:lintest}
Let $q$ be a prime number and $d\in \mathbb{N}$.
Let $f:\F_q^d\to\F_q$ be a scalar respecting function. Let $\epsilon, \delta > 0$ be parameters such that 
$\epsilon \gg \delta \gg \frac{1}{q^{1/3}}$. 
If $f$  passes $\mathcal{T}$ with probability $\varepsilon$, then  there exists an integer $r = O(1/\delta^2)$ and 
\emph{linear} functions $c_1,\ldots ,c_r:\F_q^d\to\F_q$, such that the following holds.
$$
\underset{ (\alphab, \betab) \sim S_{f,\mathcal{T}} }{\Pr}\left[\exists\ \text{unique } j\in[r] \text{ such that }f(\alphab)=c_j(\alphab),
 f(\betab)=c_j(\betab) \right]\ge 1- O\left(\frac{\delta}{\epsilon}\right).
$$
\end{theorem}

The proof of the above theorem follows by combining known ideas in literature, more precisely, we combine the arguments made in \cite{AS03} and \cite{KS13} to obtain the theorem. We include a proof of the above theorem in Appendix~\ref{sec:appendix} for the sake of completeness. 

Next, we extend the above theorem to  functions from $\F_q^d$ to $\F_q^\ell$ for any $\ell\in\mathbb{N}$ by using the above theorem as a blackbox result.  

\begin{theorem}[Piecing Together]\label{thm:piecing}
Let $q$ be a prime number and $d,\ell\in \mathbb{N}$. 
Let $f:\F_q^d\to\F_q^{\ell}$ be a scalar respecting function. 
Let $\epsilon, \kappa > 0$ be parameters such that 
$1/2\ge \kappa \ge \epsilon \gg  \frac{1}{q^{1/24}}$. 
If $f$  passes $\mathcal{T}$ with probability $\varepsilon$, then there exists a \emph{linear} function $c:\F_q^d\to\F_q^{\ell}$, such that the following holds.
$$
\underset{\alphab\sim \var(f,\mathcal{T}) }{\Pr}\left[\|f(\alphab)-c(\alphab)\|\le \kappa \right]\ge \frac{\varepsilon}{2}.
$$ 
\end{theorem} 
\begin{proof}
	The proof outline of this theorem is as follows. We consider $f_i:\F_q^d\to\F_q$, which is the restriction of $f$ to its $i^{\text{th}}$ output coordinate/symbol, and invoke Theorem~\ref{thm:lintest} to obtain a list of linear functions that agree with it uniquely on most of the pairs $(\alphab,\betab)$ that pass the linearity test. Next we argue that since linear functions agree on at most $1/q$ fraction of the domain, it is possible to assume that for most $\alphab\in \var(f,\mathcal{T})$, all its corresponding pairs in $S(f,\mathcal{T})$ are consistent on the $i^{\text{th}}$  coordinate w.r.t.\ a single linear function from the list. Then, by an averaging argument, we identify $\alphab^*\in \var(f,\mathcal{T})$ and $T_{\alphab^*}\subseteq \var(f,\mathcal{T})$ such that (i) for all $\betab\in T_{\alphab^*}$, we have that $(\alphab^*,\betab)$ agrees uniquely with a linear function from the list on most of the $\ell$ coordinates, and (ii) $|T_{\alphab^*} |=\Omega(\varepsilon\cdot  q^d)$. Thus, by stitching together the linear functions in the list that agree with $\alphab^*$ on each coordinate, we obtain the global linear function claimed in the theorem statement.  We now present the formal proof.

	 For every $i\in[\ell]$, let $f_i:\F_q^d\to\F_q$ be a scalar respecting function obtained by only looking at the $i^{\text{th}}$ coordinate output of $f$. It immediately follows that  
	$f_i$ passes $\mathcal{T}$ with probability $\varepsilon_i\ge \varepsilon$, since $S_{f,\mathcal{T}}\subseteq S_{f_i,\mathcal{T}}$. Moreover, each $f_i$ is also scalar respecting, and thus from Theorem~\ref{thm:lintest} we have that there are  $r_i = O(1/\delta_i^2)$ many linear functions $c_1^i,\ldots ,c_{r_i}^i:\F_q^d\to\F_q$,  such that the following holds.
	$$
	\underset{ (\alphab, \betab) \sim S_{f_i,\mathcal{T}} }{\Pr}\left[\exists\ \text{unique } j\in[r_i] \text{ such that }f_i(\alphab)=c_j^i(\alphab),
	f_i(\betab)=c_j^i(\betab) \right]\ge 1- O\left(\frac{\delta_i}{\epsilon_i}\right).
	$$
	
	Let $R_i\subseteq S_{f_i,\mathcal{T}}$ be the largest set such that for all $(\alphab, \betab)\in R_i$, the event `$\exists\ \text{unique } j\in[r_i] \text{ such that }f_i(\alphab)=c_j^i(\alphab),
	f_i(\betab)=c_j^i(\betab)$' does not happen. We note that $|R_i|=O\left(\frac{\delta_i}{\epsilon_i}\right)\cdot |S_{f_i,\mathcal{T}}|\le O\left(\frac{\delta_i}{\epsilon_i}\right)\cdot q^{2d}= O\left(\frac{\delta_i}{\epsilon_i}\right)\cdot \frac{|S_{f,\mathcal{T}}|}{\varepsilon}$. 
	
	We set $\delta_i=O(\varepsilon^{4}\cdot \varepsilon_i)$, to obtain a subset $\tilde{S}_{f_i,\mathcal{T}}:=  S_{f,\mathcal{T}}\setminus R_i$, for which we know $|\tilde{S}_{f_i,\mathcal{T}}|\ge |{S}_{f,\mathcal{T}}|-|R_i| \ge  (1-\varepsilon^{3})\cdot |{S}_{f,\mathcal{T}}|$. Note that  for every $(\alphab, \betab) \in \tilde{S}_{f_i,\mathcal{T}}$ there exists unique $j\in[r_i]$ such that  $f_i(\alphab)=c_j^i(\alphab)$ and $f_i(\betab)=c_j^i(\betab)$.

	For every $i\in [\ell]$, let $A_i\subseteq \var(f,\mathcal{T})$ be the largest subset for which we have that for every  $\alphab\in A_i$ there are two distinct $j,j'\in [r_i]$ such that $c_j^i(\alphab)=c_{j'}^i(\alphab)$. 
	Since every two linear functions on the domain $\F_q^d$ can agree on their evaluations on exactly $q^{d-1}$ points of the domain, we have that:
	$$
	|A_i|\le {q^{d-1}}\cdot \binom{r_i}{2}=O\left(\frac{q^{d-1}}{\delta_i^4}\right)= O\left(\frac{q^{d-1}}{\varepsilon^{16}\cdot \varepsilon_i^4}\right)= O\left(\frac{q^{d-1}}{\varepsilon^{20}}\right)\le \varepsilon^{4}\cdot q^d,
	$$
	where in the last inequality we used that $\varepsilon\gg\frac{1}{q^{1/24}}$.

	Let $G = (V,E)$ be a directed graph where $V$ is associated with $\var(f,\mathcal{T})$, and we have an edge from $\alphab$ to $\betab$ in $E$ if and only if $(\alphab,\betab)\in S_{f,\mathcal{T}}$ (i.e., $|E|=|S_{f,\mathcal{T}}|$). For every $i\in [\ell]$, let $\tilde G_i = (V,\tilde E_i)$ be a subgraph of $G$ where we have an edge from $\alphab$ to $\betab$ in $\tilde E_i$ if and only if $(\alphab,\betab)\in \tilde{S}_{f_i,\mathcal{T}}$ and $\alphab\notin A_i$. Note the following lower bound on size of $\tilde E_i$: $$|\tilde E_i|\ge |\tilde S_{f_i,\mathcal{T}}|-\left(|A_i|\cdot q^d\right)\ge  \left((1-\varepsilon^{3})\cdot \varepsilon-\varepsilon^{4}\right)\cdot q^{2d}\ge (1-\varepsilon^2)\cdot \varepsilon\cdot q^{2d}=(1-\varepsilon^2)\cdot |E|,$$ where we used that $\varepsilon\le 1/2$.
	
	We now claim the following:
	
	\begin{claim}\label{claimTh3.2}
		There is a vertex $\alphab^*$ in $V$ and a subset $T_{\alphab^*}$ of its out-neighbors in $G$ such that the following holds:
		\begin{itemize}
\item $|T_{\alphab^*}|\ge (1-\varepsilon)\cdot \varepsilon\cdot q^d$. 
\item $\forall \betab\in T_{\alphab^*},$ we have $\underset{\substack{i\sim [\ell]}}{\Pr}\left[\left(\alphab^*,\betab\right)\in \tilde E_i\right]\ge 1-\varepsilon$.
		\end{itemize}
	\end{claim}
	
	Assuming the above claim, we will first complete the proof of the theorem. We now define $s_1,\ldots ,s_{\ell}\in \mathbb{Z}$ as follows. For every  $i\in [\ell]$ if there is some $\betab \in T_{\alphab^*}$ such that $(\alphab^*,\betab)\in \tilde E_i$, then let $s_i\in [r_i]$ be the unique index such that $c_{s_i}^i(\alphab)= f_i(\alphab)$ (the uniqueness follows from the observation that $\alphab^*\notin A_i$); otherwise we define $s_i$ to be 0.
	We build a linear function $c:\F_q^d\to\F_q^\ell$ as follows. For every $i\in [\ell]$, let $c_0^i:\F_q^d\to\F_q$ denote the all zeroes function. Then,
	$$ \forall \alphab\in \F_q^d,\ 
	c(\alphab):=\left(c_{s_1}^1(\alphab),c_{s_2}^2(\alphab),\ldots ,c_{s_{\ell}}^\ell(\alphab)\right).
	$$

	For every $\betab\in T_{\alphab^\ast}$, we claim that $$
	\|f(\betab)-c(\betab)\|\le \varepsilon.
	$$
	To see this, note that for every $i\in[\ell]$ such that $(\alphab^*,\betab)\in\tilde E_i$ we have that $c_{s_i}^i$ agrees with $f_i$ on $\betab$ (since $(\alphab^*,\betab)\in\tilde S_{f_i,\mathcal{T}}$ and $\alphab^*\notin A_i$). Thus, we have 
	$$
	\|f(\betab)-c(\betab)\|\le 1-\underset{\substack{i\sim [\ell]}}{\Pr}\left[\left(\alphab^*,\betab\right)\in \tilde E_i\right]\le \varepsilon \le \tau.
	$$
	
	Finally, we have that 
	\begin{align*}
		&
		\underset{\alphab\sim \var(f,\mathcal{T}) }{\Pr}\left[\|f(\alphab)-c(\alphab)\|\le \kappa \right]\ge \frac{|T_{\alphab^\ast}|}{|\var(f,\mathcal{T})|}\ge (1-\varepsilon)\cdot \varepsilon\ge \frac{\varepsilon}{2},
	\end{align*}
where in the last inequality we used that $\varepsilon\le 1/2$.	Thus, we return to the proof of Claim~\ref{claimTh3.2}.
	
	 \begin{proof}[Proof of Claim~\ref{claimTh3.2}] Recall that for all $i\in [\ell]$, we have that 
		$|\tilde E_i|\ge (1-\varepsilon^2)\cdot |E|$. This implies that $\sum_{i\in [\ell]}|\tilde E_i|\ge (1-\varepsilon^2)\cdot |E|\ell$.
		
		Let $\tilde E\subseteq E$ be the largest subset such that for all $e\in \tilde E$, we have $\underset{i\in [\ell]}{\Pr}[e\in \tilde E_i]\ge 1-\varepsilon$.  Then we claim that $|\tilde E|\ge (1-\varepsilon)\cdot |E|$. Suppose otherwise, then there are more than $\varepsilon\cdot |E|$ many edges  such that each of them appears in less than $1-\varepsilon$ fraction of the graphs $G_1,\ldots ,G_{\ell}$. Even if the remaining edges (i.e., at most $(1-\varepsilon)\cdot |E|-1$ many edges) appear in all the graphs, this implies that $\sum_{i\in [\ell]}|\tilde E_i|< \ell\cdot( (1-\varepsilon)|E| + \varepsilon\cdot |E|\cdot (1-\varepsilon)) = \ell\cdot \left(|E| -\varepsilon^{2}\cdot |E|\right)$, a contradiction. 
		
		Let $\tilde G$ be the subgraph of $G$ whose edge set is $\tilde E$. Then by an averaging argument, there is a vertex $\alphab^*$ in $\tilde G$ whose  degree is $|\tilde E|/\var(f,\mathcal{T})\ge (1-\varepsilon)|E|/q^d= (1-\varepsilon)\cdot \varepsilon\cdot q^d$. The subset $T_{\alphab^*}$ of neighbors of $\alphab^*$ in the claim statement is exactly the neighbors of $\alphab^*$ in $\tilde{G}$ and we note that for every $\betab\in T_{\alphab^*}$, we have that  $\underset{i\in [\ell]}{\Pr}[(\alphab^*,\betab)\in \tilde E_i]\ge 1-\varepsilon$. 
	\end{proof}
	
\end{proof}

\section{Almost Polynomial Factor FPT Inapproximability of k-Clique}\label{sec:main}

In this section we prove Theorem~\ref{thm:main}. More precisely, we prove the following.
\begin{theorem}\label{thm:clique}
Let $\mathbb{P}$ be the set of all prime numbers. 
For every  increasing\footnoteref{footnote:increasing} computable function $F:\mathbb{N}\to\mathbb{N}$ such that $\forall k\in\mathbb{N},\ F(k)\le 1+\lfloor\log k\rfloor$, there exist computable functions $\Lambda:\mathbb{N}\to\mathbb{N}$ and $\hat{q}:\mathbb{N}\to\mathbb{P}$ such that the following holds. For every fixed parameter $k\in\mathbb{N}$, there is a randomized reduction running in $\Lambda(k)^{O(1)}\cdot\poly(n)$ time which given an instance $(U_1,U_2,\ldots,U_k)$  of $k$-vector sum as input, where for all $i\in[k]$ we have that $U_i$ is a collection of $n$ vectors in $\F_{\hat{q}(k)}^{O(k^2\log n)}$, outputs a graph $G$ such that the following holds. 
\begin{description}
\item[Completeness:] If there exist $\vec{u}_1\in U_1,\ldots,\vec{u}_k\in U_k$ such that
\[
\sum_{i\in[k]}\vec{u}_i=\vec{0},
\] then, there is a clique in $G$ of size exactly $\Lambda(k)$. 
\item[Soundness:] If for all $\vec{u}_1\in U_1,\ldots,\vec{u}_k\in U_k$ we have that
\[
\sum_{i\in[k]}\vec{u}_i\neq\vec{0},
\]
then, there is no  clique in $G$  of size  $\Lambda(k)^{1-\frac{1}{F(\Lambda(k))}}$.
\item[Size:] The number of vertices in $G$ is at most $\Lambda(k)\cdot \poly(n)$.
\end{description}
\end{theorem}

The proof of Theorem~\ref{thm:main} then follows by invoking the above theorem and noting the W[1]-hardness of $k$-Vector Sum problem (Theorem~\ref{thm:kvec}).

The proof outline of Theorem~\ref{thm:clique} in the subsequent subsections is as follows. In Section~\ref{sec:definition} we introduce a few definitions and results which will be useful for the design and analysis of our reduction. In Section~\ref{sec:construction}, we outline a randomized reduction from the $k$-vector sum to the $\Lambda(k)$-clique problem. In Section~\ref{sec:analysis}, we prove the completeness, soundness, and claims on the reduction parameters. %In Section~\ref{sec:derand}, we derandomize the reduction. 

\subsection{Notations and Definitions}
\label{sec:definition}

  In this subsection we introduce a few definitions and prove some basic results which will come in handy in the subsequent subsections.

For any finite field $\F$ we define the operator $\langle\cdot \rangle:\F^d\times \F^d\to \F$ (for every $d\in\mathbb{N}$) as follows. For all $\vec{a}:=(a_1,\ldots,a_d),\vec{b}:=(b_1,\ldots,b_d)\in\F^d$ we have $$
 \langle \vec{a},\vec{b}\rangle=\sum_{i\in[d]}\left(a_i\cdot b_i\right),
$$
where the sum is over $\F$. 

Next, we define an operator $\M$ which mimics matrix multiplication but by treating the matrices as vectors. Formally, for any field $\F$ and  $t,d\in\mathbb{N}$, we define $\M:\F^{d}\times \F^{t\cdot d}\to \F^t$ as follows. For all $\vec{a}\in \F^d,\vec{b}:=(\vec{b}_1,\ldots,\vec{b}_{t})\in\F^{t\cdot d}$ (where $\vec{b}_i\in\F^d$ for all $i\in[t]$) we have:
$$
\M(\vec{a},\vec{b}):=\left(\langle \vec{a},\vec{b}_1\rangle,\ldots ,\langle \vec{a},\vec{b}_t\rangle\right).
$$

We now define a linear transformation $g$ that will be useful later on. 
Let $k\in\mathbb{N}$ and $q\in\mathbb{P}$. Let $B\subseteq \F_q^m$, where $m=\Theta(k^2\log n)$ and $|B|=n$.  Let $\ell:=12\log_q n$. In the next subsection, we will fix $k$, set $q$ to be a prime depending on $k$, and use the notations $m$ and  $\ell$ as specified here.

Select $\ell$ matrices $A_1,A_2,\ldots,A_\ell\in \mathbb{F}_q^{k\times m}$ uniformly and independently at random. For every  $\vec{b}\in \mathbb{F}_q^m$, let  
\[
g(\vec{b}):=(A_1\vec{b},\cdots,A_\ell\vec{b})\in\mathbb{F}_q^{k\cdot \ell}.
\]

   Let $\tilde{B}_r\subseteq \F_q^m$ be the $r$-sumset of $B$, i.e., $$\tilde{B}_r:=\left\{\sum_{i\in[r]}\gamma_i\cdot \vec{b}_i\bigg| \gamma_1,\ldots ,\gamma_r\in\F_q,\ \text{and }\vec{b}_1,\ldots,\vec{b}_r\in B\right\}.$$
   
   We next show that if $q$ is large but only a function of $k$ (independent of $n$), then with very high probability, the relative Hamming weight of the images of all vectors in $\tilde{B}_k$ under $g$ is high. 
  
 \begin{proposition}\label{prop:collision}
Suppose $q>2^{12k}$ but $q=O_k(1)$. Then with probability at least $1-\frac{O_k(1)}{n^{k}}$,   for every  $\vec{b}\in \tilde{B}_k\setminus\{\vec{0}\}$, we have that  $\|g(\vec{b})\|\ge 2/3$. 
\end{proposition}  
\begin{proof}
For every $i\in[\ell]$, let $\vec{a}_i^1,\ldots ,\vec{a}_i^k\in\F_q^m$ be the row vectors of  $A_i$.
Fix $\vec{b}\in \tilde{B}_k\setminus\{\vec{0}\}$. 
For any fixed $i\in[\ell]$ and $j\in[k]$, we have $$
\Pr\left[\langle \vec{a}_i^j, \vec{b}\rangle\neq 0\right]=1-\frac{1}{q},
$$
where the probability is over the selection of the random matrix row $\vec{a}_i^j$. 

Next, the probability that for a fixed $\vec{b}$ we have $\|g(\vec{b})\|< 2/3$ is upper bounded by the probability that there exists a subset  $S\subseteq [\ell]\times [k]$ of size    $\ell k/3$ such  that for every $(i,j)\in S$ we have $\langle \vec{a}_i^j, \vec{b}\rangle= 0$. Therefore, 
$$\Pr\left[\|g(\vec{b})\|< \frac{2}{3}\right]\le \binom{\ell k}{|S|}\cdot q^{-\ell k/3}.$$
By union bound, the probability that for every  $\vec{b}\in \tilde{B}_k\setminus\{\vec{0}\}$, we have  $\|g(\vec{b})\|\ge\frac{2}{3}$ is at least:
\begin{align}
 1-|\tilde{B}_k|\cdot\binom{\ell k}{\ell k/3}\cdot q^{-\ell k/3}.\label{eq:calc1}
\end{align}
Note that $|\tilde{B}_k|\le (qn)^k=O_k(1)\cdot n^k$, $\binom{\ell k}{\ell k/3}\le 2^{\ell k}=n^{12k/\log q}< n$, and $q^{-\ell k/3}\le n^{-4k}$. Thus we have expression in \eqref{eq:calc1} is lower bounded by $1-\frac{O_k(1)}{n^{k}}$. 
\end{proof}

We saw above that any two vectors in $B$ disagree on most coordinates under $g$. We see below that this continues to hold even when projected to a fixed smaller subspace.

  \begin{proposition}\label{prop:collisionalphabeta}
Suppose $q>2^{12k}$ but $q=O_k(1)$. Then with probability at least $1-\frac{O_k(1)}{n}$,   for every  distinct $\vec{b}_1,\vec{b}_2\in \tilde{B}_2$, and  linealy independent $\vec{a}_1,\vec{a}_2\in\F_q^k$, we  have that  $$\left\|\M\left(\vec{a}_1,g(\vec{b}_1)\right)- \M\left(\vec{a}_2,g(\vec{b}_2)\right)\right\|\ge \frac{1}{2}.$$  
\end{proposition}  
\begin{proof}
For fixed non-zero $\vec{a}\in\F_q^k$, $i\in[\ell]$, and any $\vec{\rho}\in \F_q^m$ we have 
\begin{align}
\Pr\left[\vec{a}^T A_i=\vec{\rho}^T\right]=\frac{1}{q^m},\nonumber
\end{align}
where the probability is over the selection of the random matrix  $A_i$. Thus, for a fixed non-zero $\vec{b}\in \F_q^m$, and   any fixed $\gamma\in \F_q$ we have 
\begin{align}
\Pr\left[\langle\vec{a}^T A_i,\vec{b}\rangle=\gamma\right]=\frac{1}{q}.\label{calc3}
\end{align}

Next the probability that for fixed distinct $\vec{b}_1,\vec{b}_2\in \tilde{B}_2\setminus\{\vec{0}\}$, and fixed linearly independent $\vec{a}_1,\vec{a}_2\in\F_q^k$ we have $\left\|\M\left(\vec{a}_1,g(\vec{b}_1)\right)-\M\left(\vec{a}_2,g(\vec{b}_2)\right)\right\|< \frac{1}{2}$ is upper bounded by the probability that there exists a subset  $S\subseteq [\ell]$ of size    $\ell /2$ such  that for every $i\in S$ we have $\langle \vec{a}_1^T A_i, \vec{b}_1\rangle= \langle \vec{a}_2^T A_i, \vec{b}_2\rangle$. However, for a fixed $i\in S$, we have from \eqref{calc3} that  
$$
\Pr\left[\langle \vec{a}_1^T A_i, \vec{b}_1\rangle= \langle \vec{a}_2^T A_i, \vec{b}_2\rangle\right]=\sum_{\gamma\in\F_q}\Pr\left[\langle \vec{a}_1^T A_i, \vec{b}_1\rangle= \langle \vec{a}_2^T A_i, \vec{b}_2\rangle=\gamma\right]=\sum_{\gamma\in \F_q}\frac{1}{q^2}=\frac{1}{q},
$$
where we used the linear independence of $\vec{a}_1$ and $\vec{a}_2$ in the penultimate equality. 
Therefore we have, 
$$\Pr\left[\left\|\M\left(\vec{a}_1,g(\vec{b}_1)\right)- \M\left(\vec{a}_2,g(\vec{b}_2)\right)\right\|< \frac{1}{2}\right]\le \binom{\ell }{|S|}\cdot q^{-\ell/2}.$$
By union bound, the probability that for every    distinct $\vec{b}_1,\vec{b}_2\in \tilde{B}_2\setminus\{\vec{0}\}$, and every linearly independent $\vec{a}_1,\vec{a}_2\in\F_q^k$, we  have that the probability that $\left\|\M\left(\vec{a}_1,g(\vec{b}_1)\right)- \M\left(\vec{a}_2,g(\vec{b}_2)\right)\right\|\ge \frac{1}{2}$  is at least:
\begin{align}
   1-n^4\cdot q^{2k}\cdot\binom{\ell }{\ell /2}\cdot q^{-\ell/2}.\label{eq:calc2}
\end{align}
Note that $\binom{\ell }{\ell /2}\le 2^{\ell}=n^{12/\log q}\le n$ and $q^{-\ell /2}\le  1/n^{6}$. Thus, we have expression in \eqref{eq:calc2} is lower bounded by $1-\frac{O_k(1)}{n}$. 

Finally, we consider the case that either $\vec{b}_1$ or $\vec{b}_2$ is $\vec{0}$. Then the proposition amounts to proving that  for every   $\vec{b}\in \tilde{B}_2\setminus\{\vec{0}\}$, and  $\vec{a}\in\F_q^k\setminus\{\vec{0}\}$, we  have that  $\left\|\M\left(\vec{a},g(\vec{b})\right)\right\|\ge \frac{1}{2}$. For fixed   $\vec{b}\in \tilde{B}_2\setminus\{\vec{0}\}$ and  $\vec{a}\in\F_q^k\setminus\{\vec{0}\}$ we have the probability that $\left\|\M\left(\vec{a},g(\vec{b})\right)\right\|< \frac{1}{2}$ is upper bounded by the probability that there exists a subset  $S\subseteq [\ell]$ of size    $\ell /2$ such  that for every $i\in S$ we have $\langle \vec{a}^T A_i, \vec{b}\rangle= 0$. However, for a fixed $i\in S$, we have from \eqref{calc3} that  this probability is $1/q$.
Therefore we have, 
$$\Pr\left[\left\|\M\left(\vec{a},g(\vec{b})\right)\right\|< \frac{1}{2}\right]\le \binom{\ell }{|S|}\cdot q^{-\ell/2}.$$
By union bound, and calculations similar to the one done previously, the proof is completed. 
\end{proof}

\subsection{Construction} \label{sec:construction}

In this subsection, we provide the reduction from the $k$-Vector Sum problem to the $\Lambda(k)$-Clique problem. 

Fix $F:\mathbb{N}\to\mathbb{N}$ as in the statement of Theorem~\ref{thm:clique}.
Without loss of generality, we assume that  $F$  satisfies the following:   for all $k\in\mathbb{N}$, we have that $F(k)\le 1+\left\lfloor \frac{\log k}{\log\log k}\right\rfloor$. This is because, suppose there is an FPT algorithm   which can decide if a graph has a clique of size $k$ or no clique of size $k^{1-1/F(k)}$, then we can use the same algorithm to decide if  a graph has a clique of size $k$ or no clique of size $k^{1-1/F'(k)}$, where $F'(k):=\min\left(F(k),1+\left\lfloor \frac{\log k}{\log\log k}\right\rfloor\right)$.

 We define the functions 
$\hat{q}:\mathbb{N}\to\mathbb{P}$ and $\Lambda:\mathbb{N}\to\mathbb{N}$ as follows. For every $k\in\mathbb{N}$, we define $\hat{q}(k)$ as the smallest prime number greater than\footnote{This lower bound on the choice of $\hat{q}(k)$ is needed as we would like to use Propositions~\ref{prop:collision}~and~\ref{prop:collisionalphabeta} later in the section.} $\max\left(2^{12k},2^{2^5}\right)$. For every $k\in\mathbb{N}$, we define $\Lambda(k):=(\hat{q}(k))^{2k^2}$. Note that, 
\begin{align}
	\Lambda(k)^{-\frac{1}{F(\Lambda(k))}}=(\hat{q}(k))^{ -\frac{2k^2}{F(\hat{q}(k))^{2k^2})}}\le (\hat{q}(k))^{ -\frac{2k^2}{\frac{2\log\hat{q}(k))^{2k^2}}{\log\log \hat{q}(k))^{2k^2}}}}= 2^{-(0.5+ \log k)}\cdot 2^{-0.5\cdot \log\log \hat{q}(k)}<  \frac{1}{8k},\label{eqF}
\end{align} 
where we used that  $F(k)\le 1+\left\lfloor \frac{\log k}{\log\log k}\right\rfloor$ and that $\hat{q}(k)>2^{2^5}$.

Fix $k\in\mathbb{N}$ and let $q:=\hat{q}(k)$.
Starting from an instance $(U_1,\ldots ,U_k)$ of $k$-Vector Sum over $\F_q$ (where the vectors are $m$-dimensional for $m=\Theta(k^2\log n)$) we  construct a graph $G(V,E)$ as follows. For all $i\in[k]$, let $|U_i|=n/k$. Let $U:=U_1\cup\cdots\cup U_k$.  Recall that $\ell=12\log_q n$. We next put together  Propositions~\ref{prop:collision}~and~\ref{prop:collisionalphabeta} as follows. We sample\footnote{The usage of these sampled matrices makes our reduction randomized.} $\ell$ matrices $A_1,A_2,\ldots,A_\ell\in \mathbb{F}_q^{k\times m}$ uniformly and independently at random and  with probability at least $1-o(1)$ we have (i) $\forall\ \gamma_1,\ldots ,\gamma_k\in \F_q, \forall\ (\vec{u}_1,\ldots ,\vec{u}_k)\in U_1\times \cdots \times U_k,$ if $\underset{i\in[k]}{\sum}\gamma_i\cdot \vec{u}_i\neq \vec{0}$ then:
\begin{align}
\left\|g\left(\sum_{i\in[k]}\gamma_i\cdot \vec{u}_i\right)\right\|\ge \frac{2}{3}\ ,\label{prop1}
\end{align} 
and (ii) $\forall$ $i\in [k]$ and for every three vectors $\vec{u}_1,\vec{u}_2,\vec{u}_3\in U_i$ such that $\vec{u}_3-\vec{u}_1\neq \vec{u}_2-\vec{u}_3$, and every linearly independent $\vec{\alpha},\vec{\beta}\in \F_q^k$ we have: 
\begin{align}
\left\|\M\left(\vec{\alpha},g(\vec{u}_3-\vec{u}_1)\right)- \M\left(\vec{\beta},g(\vec{u}_2-\vec{u}_3)\right)\right\|\ge \frac{1}{2}.\label{prop2}
\end{align} 

Now we are ready to construct $G$. First we define the vertex set $V$ of $G$:
$$
V:=\left\{(\alphab,\betab,\vec{x},\vec{y})\in \F_q^{k^2}\times\F_q^{k^2}\times \F_q^{\ell}\times \F_q^{\ell}\big| \text{ if }\alphab= \betab\text{ then }\vec{x}=\vec{y}\right\}.
$$

 Next, instead of defining the edge set $E$, we will define the graph through it's non-edges. But to do so in a clean way, we need a few additional notations and definitions.

We view every $v:=(\alphab,\betab,\vec{x},\vec{y})\in V$ as a function from $\{\alphab,\betab,\alphab+\betab\}$ to $\F_q^\ell$ where we define $v(\alphab)=\vec{x}$, $v(\betab)=\vec{y}$, and $v(\alphab+\betab)=\vec{x}+\vec{y}$.

For a vertex $v=(\alphab,\betab,\vec{x},\vec{y})\in V$, we define $\var(v):=\{\alphab,\betab,\alphab+\betab\}$. Further, for any set $T\subseteq V$, we abuse notation and  define $\var(T)$ as follows:
$$
\var(T):=\underset{v\in T}{\cup}\var(v).
$$

Finally, for every $v:=(\alphab,\betab,\vec{x},\vec{y})$ and $v^{\prime}:=(\vec{\bm{\alpha^{\prime}}},\vec{\bm{\beta^{\prime}}},\vec{x^{\prime}},\vec{y^{\prime}})\in V$ we do \emph{not} have an edge between them if and only if at least one of the following conditions hold.
\begin{description}
\item[Type 1:] $\vec{\bm{\alpha}}=\vec{\bm{\alpha^{\prime}}}$ and $\vec{\bm{\beta}}=\vec{\bm{\beta^{\prime}}}$. 
\item[Type 2:] There exists $\vec{\bm \rho}\in \var(v)\cap\var(v^{\prime})$ such that $v(\vec{\bm \rho})\neq v^{\prime}(\vec{\bm \rho})$.
\item[Type 3:] There exists some $\gamma\in \F_q$ such that $\alphab=\gamma\cdot \alphab^\prime$ and $\vec{x}\neq \gamma \cdot \vec{x}'$.
\item[Type 4:] There exist some $i\in[k]$ and $\vec{\alpha}\in\F_q^k$, such that $$\vec{\bm{\alpha}}-\vec{\bm{\alpha^{\prime}}}=\vec{\alpha}\cdot \vec{\bm{e_i}}=(\underbrace{\vec{0},\ldots, \vec{0}}_{\substack{i-1\\ \text{coordinates}}},\vec{\alpha},\vec{0},\ldots, \vec{0}),$$
and for all $\vec{u}\in U_i$ we have $ \M(\vec{\alpha},g(\vec{u}))\neq\vec{x}-\vec{x^{\prime}}$. We emphasize here that we think of each coordinate as a vector in $\F_q^k$. 
\item[Type 5:] There exists some $\vec{\alpha}\in\F_q^k$, such that $$\vec{\bm{\alpha}}-\vec{\bm{\alpha^{\prime}}}=(\vec{\alpha},\ldots, \vec{\alpha}),$$
and  $\vec{x}\neq\vec{x^{\prime}}$.
\end{description}

{The intuition behind specifying these non-edges is as follows. For every $k$-tuple of vectors $\vec{\bm{u}}:=(\vec{u}_1,\ldots ,\vec{u}_k)\in U_1\times \cdots \times U_k$ we associate a unique subset of vertices  $T_{\vec{\bm{u}}}$ as follows:
$$
T_{\vec{\bm{u}}}:=\left\{\left(\alphab=(\vec{\alpha}_1,\ldots ,\vec{\alpha}_{k}),\betab=(\vec{\beta}_1,\ldots ,\vec{\beta}_{k}),\sum_{i\in[k]}\M(\vec{\alpha}_i,g(\vec{u}_i)),\sum_{i\in[k]}\M(\vec{\beta}_i,g(\vec{u}_i))\right)\bigg| \alphab,\betab\in \F_q^{k^2}\right\}.
$$
The claim then is that if $\vec{u}_1+\cdots +\vec{u}_k=\vec{0}$ then $T_{\vec{\bm{u}}}$ is a clique. On the other hand if $\vec{u}_1+\cdots +\vec{u}_k\neq\vec{0}$ then the Type 5 non-edges ensure that there is no $\frac{|T_{\vec{\bm{u}}}|}{k\cdot q^{1/k}}$ sized\footnote{\label{footnote}In fact, we could claim that if $\vec{u}_1+\cdots +\vec{u}_k\neq\vec{0}$ then there is no $|T_{\vec{\bm{u}}}|/q^{\delta}$ sized clique, for some tiny $\delta>0$. } clique in the graph induced by $T_{\vec{\bm{u}}}$. 

On the other hand if we pick any subset $T^\prime\subseteq V$ of size $q^{2k^2}$ in $G$ then one of the first four types of non-edges ensures that there is no $\frac{|T^\prime|}{k\cdot q^{1/k}}$ sized clique in the graph induced by $T^\prime$. In other words, the first four types of non-edges incentivize to pick subset of vertices which corresponds to $T_{\vec{\bm{u}}}$ for some $\vec{\bm{u}}\in U_1\times \cdots \times U_k$. 
Type 1 non-edges incentivize to  include only one vertex in $T^\prime$ of the form $(\alphab,\betab,\vec{x},\vec{y})$ for every $\alphab,\betab\in \F_q^{k^2}$. Type 2 non-edges incentivize only to pick those vertices which are "consistent", i.e., we can extract an assignment $\sigma:\var(T^\prime)\to\F_q^\ell$ in a consistent manner.  Type 3 non-edges are introduced for technical reasons, as we would like to invoke Theorem~\ref{thm:piecing} in our analysis, i.e., to say that if $T^\prime$ contains a large clique, then it must have some "linear structure". Equipped with having an assignment $\sigma$ and some linear structure, the dearth of Type 4 non-edges enables us to decode a vector $\vec{u}_i^\ast\in U_i$ such that $T^\prime$ has a large intersection with $T_{\vec{\bm{u}}^\ast}$  where
$\vec{\bm{u}}^\ast:=(\vec{u}_1^\ast,\ldots ,\vec{u}_k^\ast)\in U_1\times \cdots \times U_k$. }

In summary, Types 1-4 non-edges ensure that any subset $T\subseteq V$ of size $q^{2k^2}$ in $G$ which contains a large clique must overlap significantly with  $T_{\vec{\bm{u}}}$ for some $\vec{\bm{u}}\in U_1\times \cdots \times U_k$. Then the lack of Type 5 non-edges ensure that if $T$ has a large clique then the $k$-tuple of vectors represented by $\vec{\bm{u}}$ must sum to $\vec{0}$.

\subsection{Analysis}\label{sec:analysis}
In this section, we analyze the parameters of the reduction, and prove the completeness and soundness claims of the theorem statement.

\paragraph{Parameters of the reduction.} 
The new graph has at most  $|\F_q^{2k^2}|\cdot |\F_q^{2\ell}|=\Lambda(k)\cdot n^{24}$ many vertices. The time needed to construct this graph is $\Lambda(k)\cdot n^{25}$.  

\paragraph{Completeness.} Suppose there exist $\vec{u}_1\in U_1,\ldots,\vec{u}_k\in U_k$ such that
$
\sum_{i\in[k]}\vec{u}_i=\vec{0}
$. Then, we can find a clique of size $|\F_q|^{2k^2}$ in $G$ as follows. Consider $T\subseteq V$ defined as below:
$$
T:=\left\{\left(\vec{\alpha_1},\ldots ,\vec{\alpha_{k}},\vec{\beta_1},\ldots ,\vec{\beta_{k}},\sum_{i\in[k]}\M(\vec{\alpha_i},g(\vec{u_i})),\sum_{i\in[k]}\M(\vec{\beta_i},g(\vec{u_i}))\right)\bigg| \vec{\alpha_i},\vec{\beta_i}\in\F_q^k, i\in[k]\right\}.
$$ 
We claim that every pair of distinct vertices in $T$ have an edge in $G$ and since $|T|=q^{2k^2}$, the completeness case follows.  

First note that if we fix any $\vec{\alpha_1},\ldots ,\vec{\alpha_{k}},\vec{\beta_1},\ldots ,\vec{\beta_{k}}\in \F_q^k$ then there are unique vectors $\vec{x},\vec{y}\in\F_q^\ell$ such that  $\left(\vec{\alpha_1},\ldots ,\vec{\alpha_{k}},\vec{\beta_1},\ldots ,\vec{\beta_{k}},\vec{x},\vec{y}\right)$ is in $T$. Thus, there are no Type 1 non-edges in subgraph induced by $T$.

Next, for every two distinct vertices $v,v^{\prime}\in T$, and for every $\vec{\bm{\rho}}:=(\vec{\rho_1},\ldots ,\vec{\rho_{k}})\in \var(v)\cap\var(v^{\prime})$, we have  $$
v(\vec{\bm{\rho}})=v^{\prime}(\vec{\bm{\rho}})=\sum_{i\in[k]} \M(\vec{\rho_i},g(\vec{u_i})), 
$$
and thus there are no Type 2 non-edges in subgraph induced by $T$.

Then,  we note that there are no Type 3 non-edges in the subgraph induced by $T$ because for every $v:=(\alphab,\betab,\vec{x},\vec{y})\in T$ and every $\gamma\in F_q$, if $v^\prime:=(\gamma\cdot \alphab,\betab^\prime,\vec{x}^\prime,\vec{y}^\prime)\in T$, then we have:
$$
\gamma\cdot \vec{x}=\gamma \cdot \sum_{i\in[k]}\M(\vec{\alpha_i},g(\vec{u_i})) = \sum_{i\in[k]}\M(\gamma \cdot \vec{\alpha_i},g(\vec{u_i})=\vec{x}'.
$$

In order to next show that there are no Type 4 non-edges in subgraph induced by $T$, we first fix $v:=(\alphab,\betab,\vec{x},\vec{y})\in T$, $i\in [k]$, and $\vec{\alpha}\in \F_q^k$. Suppose there exists $v^{\prime}:=(\alphab-\vec{\alpha}\cdot \vec{\bm{e_i}},\vec{\bm{\beta^{\prime}}},\vec{x^{\prime}},\vec{y^{\prime}})\in T$. Then we have 
\begin{align*}
\vec{x}-\vec{x^{\prime}}&=\sum_{j\in[k]} \M(\vec{\alpha_j},g(\vec{u_j}))-\left(\sum_{\substack{j\in[k]\\j\neq i}} \M(\vec{\alpha_j},g(\vec{u_j}))\right)-\M(\vec{\alpha_i}-\vec{\alpha},g(\vec{u}_i))\\
&=\M(\vec{\alpha_i},g(\vec{u}_i))-\M(\vec{\alpha_i}-\vec{\alpha},g(\vec{u}_i))\\
&=\M(\vec{\alpha},g(\vec{u}_i)).
\end{align*} 
Thus, $(v,v^{\prime})$ is an edge in the subgraph induced by $T$. 

Finally, we show that there are no Type 5 non-edges in subgraph induced by $T$. Let $v:=(\alphab,\betab,\vec{x},\vec{y})\in T$ and $\vec{\alpha}\in \F_q^k$. Suppose there exists $v^{\prime}:=(\alphab-(\vec{\alpha},\ldots ,\vec{\alpha}),\vec{\bm{\beta^{\prime}}},\vec{x^{\prime}},\vec{y^{\prime}})\in T$. Then we have 
\begin{align*}
\vec{x}-\vec{x^{\prime}}&=\sum_{i\in[k]} \M(\vec{\alpha_i},g(\vec{u_i}))-\sum_{i\in[k]} \M(\vec{\alpha_i}-\vec{\alpha},g(\vec{u_i}))\\
&=\sum_{i\in[k]} \M(\vec{\alpha},g(\vec{u_i}))
= \M\left(\vec{\alpha},g\left(\sum_{i\in[k]}\vec{u_i}\right)\right)
= \M(\vec{\alpha},g(\vec{0}))
= \vec{0}.
\end{align*} 
Thus, $(v,v^{\prime})$ is an edge in the subgraph induced by $T$. 

\paragraph{Soundness.}
Let $T$ be the set of vertices of the largest clique in $G$ (breaking ties arbitrarily). 
Let $\varepsilon:=\frac{1}{q^{{2k^2}/{F\left(q^{2k^2}\right)}}}$. 
Suppose  $|T|\ge \varepsilon\cdot q^{2k^2}$, then we shall show that for every $i\in[k]$ there exists $u_i^\ast\in U_i$ such that $u_1^\ast+\cdots +u_k^\ast=\vec{0}$.

The proof strategy is as follows. First, using $T$, we construct a function $\Gamma$ from $\F_q^{k^2}$ to $\F_q^\ell$ which we show passes the linearity test with probability at least $\varepsilon$ (Claim~\ref{claim:lintest}). Then, we invoke Theorem~\ref{thm:piecing} to say that there exists a collection of few linear functions on $k$ variables over $\F_q^k$ with coefficients from $\F_q^{k\times \ell}$ with the following property: for many queries $(\vec{\bm \alpha},\vec{\bm\beta})\in \F_q^{ k^2}\times \F_q^{ k^2}$ on which $\Gamma$ passes the linearity test, we have   a fixed linear function in our collection whose evaluation on $\vec{\bm\alpha}$ agrees with  $\Gamma(\vec{\bm\alpha})$.

Then for every $i\in[k]$ and $\vec{\alpha}\in\F_q^k$, we will identify  $\vec{u}_{\vec{\alpha}}\in U_i$ such that $\M(\vec{\alpha},g(\vec{u}_{\vec{\alpha}}))$ is roughly equal to evaluating the linear function at $\vec{\bm e}_i\cdot \vec{\alpha}$ (Claim~\ref{claim:decodespecificalpha}). Next, we show that for every $i\in[k]$, there is a single $\vec{u}_i\in U_i$  such that for all $\vec{\alpha}\in \F_q^k$ we have that  $\M(\vec{\alpha},g(\vec{u}_i))$ is roughly equal to evaluating the linear function at $\vec{\bm e}_i\cdot \vec{\alpha}$ (Claim~\ref{claim:decodei}). Finally, the proof follows by  observing that there are no Type 5 non-edges in $T$, and thus   these identified $\vec{u}_i$s must sum to $\vec{0}$.

We now begin the formal soundness case analysis. We claim that for every $\alphab\in \var(T)$, if there exist distinct $v,v^{\prime}\in T$ such that $\alphab\in \var(v)\cap\var(v^{\prime})$ then, $v(\alphab)=v^{\prime}(\alphab)$. Otherwise, $(v,v^{\prime})$ would be a non-edge of Type 2 which is not possible as the vertices in $T$ form a clique. %Thus, we may construct a function $\sigma:\var(T)\to\F_q^\ell$ where for every $\alphab\in\var(T)$, we define $\sigma(\alphab)=v(\alphab)$, where $v\in T$ such that    $\alphab\in\var(v)$.\cs{Maybe I don't really need the extra notation $\sigma$.}

We construct a function $\Gamma:\F_q^{k^2}\to \F_q^\ell$  in two phases. In the first phase, we define $\Gamma$ only for vectors in $\var(T)$.   For every $\vec{\bm{\alpha}}\in\var(T)$, we set $
\Gamma(\vec{\bm{\alpha}})=
v(\vec{\bm{\alpha}})\text{ if }v\in T\text{ is such that }\vec{\bm{\alpha}}\in \var(v)$.
In the second phase, 
we iteratively go over all the vectors in $\F_q^{k^2}\setminus \var(T)$ in some canonical order. In the $j^{\text{th}}$ iteration, let $\alphab_j$ be the vector considered. If there exist $\gamma\in \F_q$ and $\alphab'\in \var(T)$ such that $\alphab_j=\gamma\cdot \alphab'$ then we define $\Gamma(\alphab_j)=\gamma\cdot \Gamma(\alphab')$; otherwise if there exists $j'<j$ such that $\alphab_j=\gamma\cdot \alphab_{j'}$ for some $\gamma\in\F_q$ then we define $\Gamma(\alphab_j)=\gamma\cdot \Gamma(\alphab_{j'})$; otherwise, we set $\Gamma(\alphab_j)$ to a uniformly random vector in $\F_q^\ell$. 

Notice that by our construction and that there are no Type 3 non-edges in $T$, we have that for all $\alphab\in \F_q^{k^2}$ and for all $\gamma\in \F_q$ we have $\Gamma(\gamma\cdot\alphab)=\gamma\cdot \Gamma(\alphab)$, i.e., $\Gamma$ is scalar respecting. 

Next, we have the following claim on $\Gamma$ passing the linearity test.

\begin{claim}\label{claim:lintest}
$\Gamma$ passes the linearity test with probability at least $\varepsilon$. 
\end{claim}

To see the claim, first consider the set $S\subseteq \F_q^{k^2}\times \F_q^{k^2}$ defined as follows.
$$
S:=\underset{(\alphab,\betab,\vec{x},\vec{y})\in T}{\bigcup}\{(\alphab,\betab)\}.
$$  

Notice that the probability of $\Gamma$ passing the linearity test is lower bounded by:
$$
\frac{|S|}{|\F|^{2k^2}}\cdot \Pr_{(\vec{\bm{\alpha}},\vec{\bm{\beta}})\sim S}\left[\Gamma(\vec{\bm{\alpha}})+\Gamma(\vec{\bm{\beta}})=\Gamma(\vec{\bm{\alpha}}+\vec{\bm{\beta}})\right].
$$
 
However, for any $(\alphab,\betab)\in S$, we have that $(\alphab,\betab,\Gamma(\alphab),\Gamma(\betab))$ is in $T$ by construction of set $S$.  Thus, we have $\Pr_{(\vec{\bm{\alpha}},\vec{\bm{\beta}})\sim S}\left[\Gamma(\vec{\bm{\alpha}})+\Gamma(\vec{\bm{\beta}})=\Gamma(\vec{\bm{\alpha}}+\vec{\bm{\beta}})\right]=1$ and $\Gamma$ passes the linearity test with probability at least $|S|/|\F_q|^{2k^2}$. Since $|S|=|T|$, we have that the proof of Claim~\ref{claim:lintest} is completed.

Next invoking Theorem~\ref{thm:piecing} with $\kappa=\frac{1}{8k}$ (since $\Gamma$ is scalar respecting and $\kappa> \varepsilon$ from \eqref{eqF}), we have that there exist a linear function $c:\F_q^{k^2}\to\F_q^{\ell}$, such that the following holds.
$$
\underset{\vec{\bm{\alpha}}\sim \var(T) }{\Pr}\left[\|\Gamma(\vec{\bm{\alpha}})-c(\vec{\bm{\alpha}})\|\le \kappa  \right]\ge \frac{\varepsilon}{2}.
$$ 

Let $R^\ast\subseteq \var(T)$ be the largest sized subset such that the following holds:
\begin{align}
\underset{\vec{\bm{\alpha}}\sim R^\ast }{\Pr}\left[\|\Gamma(\vec{\bm{\alpha}})-c(\vec{\bm{\alpha}})\|\le \kappa  \right]=1
.\label{eq:lin1}
\end{align}

We note that $|R^\ast|\ge \frac{\varepsilon}{2}\cdot |\var(T)|$, and since $|\var(T)|\ge \varepsilon\cdot |\F_q^{k^2}|$, we have that $|R^\ast|\ge \frac{\varepsilon^2}{2}\cdot q^{k^2}$.  
Next, we think of $c$ as a linear function on $k$ variables over $\F_q^k$ with coefficients  in $\F_q^{k\times \ell}$:
$$c(\vec{\alpha}_1,\ldots ,\vec{\alpha}_k)=\sum_{i\in[k]}\M(\vec{\alpha}_i,\vec{\Theta}_i ),$$ 
for some $\vec{\Theta}_1 ,\ldots ,\vec{\Theta}_k \in\mathbb{F}^{k\times \ell}$.

\begin{claim}\label{claim:decodespecificalpha}
For every $i\in[k]$ and $\vec{\alpha}\in\F_q^k$, there exists $\vec{u}_i^\ast\in U_i$  such that $\|\M(  \vec{\alpha},\vec{\Theta}_i-g(\vec{u}_i^\ast)) \|\le 2\kappa$.
\end{claim}

%Fix $\vec{\alpha}\in\F_q^k$. We say $\vec{\alpha}$ is useful if there exists $(\vec{\bm{\alpha}},\vec{\bm{\beta}})\in R^\ast$ such that for some $\gamma_{\vec{\alpha}}\in \F$, we have that  $(\vec{\bm{\alpha}}+\vec{\bm e_i}\cdot (\gamma\cdot\vec{\alpha}_{\vec{\alpha}}),\vec{\bm{\beta}})\in R^\ast$ too.

If $\vec{\alpha}=\vec{0}$, the claim trivially holds. Therefore we assume that $\vec{\alpha}\in\F_q^k\setminus \{\vec{0}\}$. For every $i\in[k]$ and every $\vec{\alpha}\in \F_q^k$, we show that there is a line in the direction of $\vec{\alpha}\cdot \vec{\bm{e_i}}$ which contains two vertices $(\alphab,\betab,\vec{x},\vec{y})$ and  $(\alphab',\betab',\vec{x}',\vec{y}')$ such that  $\alphab,\alphab'\in R^\ast$. Then, by noting that these two vertices don't have a Type 4 non-edge between them, we identify $\vec{u}_i^\ast\in U_i$. The formal argument follows.

We say a line is linear if it passes through the origin and affine otherwise. Note that every linear line can be identified through one of the non-zero points on it. Also note that for every linear line in $\F_q^{k^2}$, the line   and all its affine shifts always cover the entire space $\F_q^{k^2}$ . 
 Since $|{R}^\ast|/q^{k^2}>\varepsilon^2/2>1/q$, we have that  by an averaging argument,  for every linear line, either that line or one of it's affine shifts   contains at least two points in $R^\ast$. We use this argument below and in the proof of Claim~\ref{claim:testzero}.

Fix $i\in[k]$ and $\vec{\alpha}\in\F_q^k\setminus \{\vec{0}\}$. Let $L_i$ be a linear line in $\F_q^{k^2}$  containing the point  $\vec{\alpha}\cdot \vec{\bm e}_i$. Then there exist two points  $\vec{\bm{\alpha}},\vec{\bm{\alpha}}+\vec{\bm e}_i\cdot (\gamma\cdot\vec{\alpha})\in  {R}^\ast$, for some $\alphab\in\F_q^{k^2}$ and $\gamma\in \F_q\setminus\{0\}$. From \eqref{eq:lin1} we have
\begin{align}
\|\Gamma(\vec{\bm{\alpha}})-c(\vec{\bm{\alpha}})\|\le \kappa\ \text{and } \|\Gamma(\vec{\bm{\alpha}}+\vec{\bm e}_i\cdot (\gamma\cdot\vec{\alpha}))-c(\vec{\bm{\alpha}}+\vec{\bm e}_i\cdot (\gamma\cdot\vec{\alpha}))\|\le \kappa. 
\label{eq:lin2}
\end{align}

Let $v:=(\alphab,\betab,\Gamma(\alphab),\Gamma(\betab))$ and 
$v^{\prime}:=(\alphab+\vec{\bm e}_i\cdot (\gamma\cdot\vec{\alpha}),\betab',\Gamma(\alphab+\vec{\bm e}_i\cdot (\gamma\cdot\vec{\alpha})),\Gamma(\betab'))$  be the two vertices in $T$ for some $\betab,\betab'\in \F_q^{k^2}$. Since there  is no Type 4 non-edge between them,  there exists $u_i^\ast\in U_i$ such that  
\begin{align}
\Gamma(\alphab+\vec{\bm e}_i\cdot (\gamma\cdot\vec{\alpha}))-\Gamma(\alphab)=\M(\gamma\cdot\vec{\alpha},g(\vec{u}_i^\ast)).\label{eq:lin4}
\end{align}

On a different note, we have
\begin{align}
c(\vec{\bm{\alpha}}+\vec{\bm e}_i\cdot (\gamma\cdot\vec{\alpha}))=c(\vec{\bm{\alpha}})+\gamma\cdot c(\vec{\bm e}_i\cdot \vec{\alpha})=c(\vec{\bm{\alpha}})+\M(\gamma\cdot \vec{\alpha},\vec{\Theta}_i).
\label{eq:lin3}
\end{align}

Plugging in the simplification in \eqref{eq:lin4} and \eqref{eq:lin3} into \eqref{eq:lin2}, we have
\allowdisplaybreaks
\begin{align*}
2\kappa&\ge\|\Gamma(\vec{\bm{\alpha}})-c(\vec{\bm{\alpha}})\|+\|\Gamma(\vec{\bm{\alpha}}+\vec{\bm e}_i\cdot (\gamma\cdot\vec{\alpha}))-c(\vec{\bm{\alpha}}+\vec{\bm e}_i\cdot (\gamma\cdot\vec{\alpha}))\|\\
&\ge\|\Gamma(\vec{\bm{\alpha}})-c(\vec{\bm{\alpha}})-\Gamma(\vec{\bm{\alpha}}+\vec{\bm e}_i\cdot (\gamma\cdot\vec{\alpha}))+c(\vec{\bm{\alpha}}+\vec{\bm e}_i\cdot (\gamma\cdot\vec{\alpha}))\|\\
&=\|\M(\gamma\cdot \vec{\alpha},\vec{\Theta}_i)-\M(\gamma\cdot\vec{\alpha},g(\vec{u_i^\ast})) \|\\
&=\|\M(\gamma\cdot \vec{\alpha},\vec{\Theta}_i-g(\vec{u_i^\ast})) \|\\
&=\|\M(  \vec{\alpha},\vec{\Theta}_i-g(\vec{u_i^\ast})) \|,
\end{align*}
 where the last equality follows from noting that for any vector $\vec{a}$ and non-zero scalar $\zeta$, we have $\|\zeta\cdot \vec{a}\|=\|\vec{a}\|$.

\begin{claim}\label{claim:decodei}
For every $i\in[k]$, there exists $\vec{u}_i^\ast\in U_i$  such that  for every $\vec{\alpha}\in\F_q^k$ we have $$ \|\M(  \vec{\alpha},\vec{\Theta}_i-g(\vec{u_i^\ast})) \|\le 2\kappa.$$ 
\end{claim}

We prove the claim for non-zero $\vec{\alpha}$ as the claim is trivial for the case $\vec{\alpha}=\vec{0}$.

For every $\vec{\alpha}\in\F_q^k\setminus \{\vec{0}\}$, let $\vec{u}_{\vec{\alpha}}\in U_i$ be the vector guaranteed in Claim~\ref{claim:decodespecificalpha}, i.e., 
$\|\M(  \vec{\alpha},\vec{\Theta}_i-g(\vec{u}_{\vec{\alpha}})) \|\le 2\kappa$. 

Now consider any linearly independent $\vec{\alpha},\vec{\beta}\in\F_q^k$. We then have: 
 $$
 \|\M(  \vec{\alpha},\vec{\Theta}_i-g(\vec{u}_{\vec{\alpha}})) \|\le 2\kappa,\ \|\M(  \vec{\beta},\vec{\Theta}_i-g(\vec{u}_{\vec{\beta}})) \|\le 2\kappa,\text{ and }\|\M(  \vec{\alpha}+\vec{\beta},\vec{\Theta}_i-g(\vec{u}_{\vec{\alpha}+\vec{\beta}})) \|\le 2\kappa.
 $$

Putting these three inequalities together:
\begin{align*}
 6\kappa&\ge \|\M(  \vec{\alpha},\vec{\Theta}_i-g(\vec{u}_{\vec{\alpha}})) \|+\|\M(  \vec{\beta},\vec{\Theta}_i-g(\vec{u}_{\vec{\beta}})) \|+\|\M(  \vec{\alpha}+\vec{\beta},\vec{\Theta}_i-g(\vec{u}_{\vec{\alpha}+\vec{\beta}})) \|\\
  &\ge\|\M(  \vec{\alpha},\vec{\Theta}_i-g(\vec{u}_{\vec{\alpha}})) +\M(  \vec{\beta},\vec{\Theta}_i-g(\vec{u}_{\vec{\beta}})) -\M(  \vec{\alpha}+\vec{\beta},\vec{\Theta}_i-g(\vec{u}_{\vec{\alpha}+\vec{\beta}})) \|\\
  &=\|\M(  \vec{\alpha},g(\vec{u}_{\vec{\alpha}})) +\M(  \vec{\beta},g(\vec{u}_{\vec{\beta}})) -\M(  \vec{\alpha}+\vec{\beta},g(\vec{u}_{\vec{\alpha}+\vec{\beta}})) \|\\
    &=\|\M(  \vec{\alpha},g(\vec{u}_{\vec{\alpha}})-g(\vec{u}_{\vec{\alpha}+\vec{\beta}})) +\M(  \vec{\beta},g(\vec{u}_{\vec{\beta}})-g(\vec{u}_{\vec{\alpha}+\vec{\beta}})) \|\\
      &=\|\M(  \vec{\alpha},g(\vec{w})) -\M(  \vec{\beta},g(\vec{w}')) \|,
\end{align*}
where $\vec{w}:=\vec{u}_{\vec{\alpha}}-\vec{u}_{\vec{\alpha}+\vec{\beta}}$ and $\vec{w}':=\vec{u}_{\vec{\alpha}+\vec{\beta}}-\vec{u}_{\vec{\beta}}$.

If $\vec{w}\neq\vec{w}'$ then  we arrive at a  contradiction  to \eqref{prop2} (since $\|\M(  \vec{\alpha},g(\vec{w}))-\M(  \vec{\beta},g(\vec{w}')) \|\le 6\kappa$ and $\kappa=o(1)$).

Thus $\vec{w}=\vec{w}'$ which implies $\vec{u}_{\vec{\alpha}}+\vec{u}_{\vec{\beta}}=2\vec{u}_{\vec{\alpha}+\vec{\beta}}$. Since the choice of $\vec{\alpha}$ and $\vec{\beta}$ were arbitrary  linearly independent  vectors, we also have:
\begin{align*}
\vec{u}_{\vec{\alpha}+\vec{\beta}}+\vec{u}_{\vec{\beta}}&=2\vec{u}_{\vec{\alpha}+2\vec{\beta}},\\ \vec{u}_{\vec{\alpha}}+\vec{u}_{\vec{\alpha}+\vec{\beta}}&=2\vec{u}_{2\vec{\alpha}+\vec{\beta}},
\\ \vec{u}_{2\vec{\alpha}+\vec{\beta}}+\vec{u}_{\vec{\beta}}&=2\vec{u}_{\vec{2\alpha}+2\vec{\beta}}=\vec{u}_{\vec{\alpha}}+\vec{u}_{\vec{\alpha}+2\vec{\beta}}.
\end{align*}

We put these relationships together to obtain the following:
\begin{align*}
\vec{u}_{\vec{\alpha}}&=\vec{u}_{\vec{\alpha}}+4\vec{u}_{\vec{2\alpha}+2\vec{\beta}}-4\vec{u}_{2\vec{\alpha}+2\vec{\beta}}\\
&=\vec{u}_{\vec{\alpha}}+\left(2\vec{u}_{2\vec{\alpha}+\vec{\beta}}+2\vec{u}_{\vec{\beta}}\right)-\left(2\vec{u}_{\vec{\alpha}}+2\vec{u}_{\vec{\alpha}+2\vec{\beta}}\right)\\
&=2\vec{u}_{2\vec{\alpha}+\vec{\beta}}+2\vec{u}_{\vec{\beta}}-2\vec{u}_{\vec{\alpha}+2\vec{\beta}}-\vec{u}_{\vec{\alpha}}\\
&=\left(\vec{u}_{\vec{\alpha}}+\vec{u}_{\vec{\alpha}+\vec{\beta}}\right)+2\vec{u}_{\vec{\beta}}-\left(\vec{u}_{\vec{\alpha}+\vec{\beta}}+\vec{u}_{\vec{\beta}}\right)-\vec{u}_{\vec{\alpha}}\\
&=\vec{u}_{\vec{\beta}}.
\end{align*}

So we are only left to handle the cases when $\vec{\alpha}$ and $\vec{\beta}$ are linearly dependent, i.e., for some $\gamma\in \F_q\setminus \{0\}$ we have $\vec{\alpha}=\gamma\cdot \vec{\beta}$. In this case let $\vec{\beta}'\in \F_q^k$  such that it is linearly independent to $\vec{\beta}$ (and thus linearly independent to $\vec{\alpha}$ as well). From the above argument we have that $\vec{u}_{\vec{\alpha}}=\vec{u}_{\vec{\beta}'}=\vec{u}_{\vec{\beta}}$.

\begin{claim}\label{claim:testzero}
We have $\vec{u}_1^\ast+\cdots +\vec{u}_k^\ast=\vec{0}$, where for all $i\in[k]$, $\vec{u}_i^\ast$ is the vector identified in Claim~\ref{claim:decodei}. 
\end{claim}

The proof idea of this claim is as follows. For every $i\in[k]$ and every $\vec{\alpha}\in \F_q^k$, we show that there is a line in the direction of $(\vec{\alpha},\ldots ,\vec{\alpha})$ which contains two vertices $(\alphab,\betab,\vec{x},\vec{y})$ and  $(\alphab',\betab',\vec{x}',\vec{y}')$ such that  $\alphab,\alphab'\in R^\ast$. Then, by noting that these two vertices don't have a Type 5 non-edge between them, we obtain that the linear function $c$ evaluated at $(\vec{\alpha},\ldots ,\vec{\alpha})$ is almost $\vec{0}$. On the other hand, from Claim~\ref{claim:decodei}, we have that $\M\left(\vec{\alpha},\underset{i\in[k]}{\sum}g(\vec{u}_i^\ast)\right)$ is close to $c(\vec{\alpha},\ldots ,\vec{\alpha})$. Thus, we obtain that $\M\left(\vec{\alpha},\underset{i\in[k]}{\sum}g(\vec{u}_i^\ast)\right)$ has small relative Hamming weight for all $\vec{\alpha}\in \F_q^k$. However, from  \eqref{prop1} we know that if $\underset{i\in[k]}{\sum}\vec{u}_i^\ast\neq \vec{0}$ then there exists $\vec{\alpha}\in\F_q^k$ such that  $\M\left(\vec{\alpha},\underset{i\in[k]}{\sum}g(\vec{u}_i^\ast)\right)$ has large relative Hamming weight, and thus, we arrive at a contradiction. The formal argument follows.

Fix some non-zero $\vec{\alpha}\in \F_q^k$.   Let $L_0$ be a linear line in $\F_q^{k^2}$  containing   the point $(\vec{\alpha},\ldots ,\vec{\alpha})$. Since $|{R}^\ast|/q^{k^2}>1/q$,  there exist two points  $\vec{\bm{\alpha}},\vec{\bm{\alpha}}+(\gamma\cdot \vec{\alpha},\ldots ,\gamma\cdot\vec{\alpha})\in {R}^\ast$, for some $\alphab\in\F_q^{k^2}$ and $\gamma\in \F_q\setminus\{0\}$. 
  From \eqref{eq:lin1} we have
\begin{align}
\|\Gamma(\vec{\bm{\alpha}})-c(\vec{\bm{\alpha}})\|\le \kappa\ \text{and } \|\Gamma(\vec{\bm{\alpha}}+(\gamma\cdot \vec{\alpha},\ldots ,\gamma\cdot\vec{\alpha}))-c(\vec{\bm{\alpha}}+(\gamma\cdot \vec{\alpha},\ldots ,\gamma\cdot\vec{\alpha}))\|\le \kappa. 
\label{eq:lin7}
\end{align}

Let $v:=(\alphab,\betab,\Gamma(\alphab),\Gamma(\betab))$ and 
$v^{\prime}:=(\alphab+(\vec{\alpha}\cdot\gamma,\ldots ,\vec{\alpha}\cdot\gamma),\betab',\Gamma(\alphab+(\vec{\alpha}\cdot\gamma,\ldots ,\vec{\alpha}\cdot\gamma)),\Gamma(\betab'))$ be the two vertices in $T$ for some $\betab,\betab'\in\F_q^{k^2}$. Since   there is no Type 5 non-edge between them, we have 
\begin{align}
\Gamma(\alphab)=\Gamma(\alphab+(\gamma\cdot \vec{\alpha},\ldots ,\gamma\cdot\vec{\alpha})).\label{eq:lin8}
\end{align}

On a different note, we have
\begin{align}
c(\vec{\bm{\alpha}}+(\gamma\cdot \vec{\alpha},\ldots ,\gamma\cdot\vec{\alpha}))=c(\vec{\bm{\alpha}})+\gamma\cdot c(\vec{\alpha},\ldots ,\vec{\alpha})=c(\vec{\bm{\alpha}})+\M\left(\gamma\cdot \vec{\alpha},\underset{i\in[k]}{\sum}\vec{\Theta}_i\right).
\label{eq:lin9}
\end{align}

Plugging in the simplification in \eqref{eq:lin8} and \eqref{eq:lin9} into \eqref{eq:lin7}, we have
\allowdisplaybreaks
\begin{align}
2\kappa&\ge\|\Gamma(\vec{\bm{\alpha}})-c(\vec{\bm{\alpha}})\|+\|\Gamma(\vec{\bm{\alpha}}+(\gamma\cdot \vec{\alpha},\ldots ,\gamma\cdot\vec{\alpha}))-c(\vec{\bm{\alpha}}+(\vec{\alpha}\cdot\gamma,\ldots ,\vec{\alpha}\cdot\gamma))\|\nonumber\\
&\ge\|\Gamma(\vec{\bm{\alpha}})-c(\vec{\bm{\alpha}})-\Gamma(\vec{\bm{\alpha}}+(\gamma\cdot \vec{\alpha},\ldots ,\gamma\cdot\vec{\alpha}))+c(\vec{\bm{\alpha}}+(\vec{\alpha}\cdot\gamma,\ldots ,\vec{\alpha}\cdot\gamma))\|\nonumber\\
&=\left\|\M\left(\gamma\cdot \vec{\alpha},\underset{i\in[k]}{\sum}\vec{\Theta}_i\right) \right\|\nonumber\\
&=\left\|\M\left( \vec{\alpha},\underset{i\in[k]}{\sum}\vec{\Theta}_i\right)\right\|\label{eqzero},
\end{align}
 where the last equality follows from noting that for any vector $\vec{a}$ and non-zero scalar $\zeta$, we have $\|\zeta\cdot \vec{a}\|=\|\vec{a}\|$.
 
Next, to see the claim, we first define $\vec{z}^\ast\in \F_q^m$ as follows:
$$
\vec{z}^\ast:=\sum_{i\in[k]}\vec{u}_i^\ast.
$$

From Claim~\ref{claim:decodei}, we have that for every $i\in[k]$ and for all $\vec{\alpha}\in \F_q^k$ we have $\|\M(\vec{\alpha},\vec{\Theta}_i-g(\vec{u}_i^\ast))\|\le 2\kappa$.  
Fix some $\vec{\alpha}\in \F_q^k\setminus \{0\}$. Then,
$$2\kappa k\ge   \sum_{i\in[k]}\left\|\M(  \vec{\alpha},\vec{\Theta}_i-g(\vec{u_i^\ast})) \right\| \ge \left\|\sum_{i\in[k]}\M(  \vec{\alpha},\vec{\Theta}_i-g(\vec{u_i^\ast})) \right\|=\left\|\M\left(  \vec{\alpha},\sum_{i\in[k]}\left(\vec{\Theta}_i-g(\vec{u_i^\ast})\right)\right) \right\|.$$ 
 
Plugging in \eqref{eqzero}, we have
$$\frac{1}{2}\ge 2\kappa (k+1)\ge \left\|\M\left(  \vec{\alpha},\sum_{i\in[k]}g(\vec{u_i^\ast})\right) \right\|=\left\|\M\left(  \vec{\alpha},g\left(\sum_{i\in[k]}\vec{u_i^\ast}\right)\right) \right\|=\left\|\M\left(  \vec{\alpha},g\left(\vec{z}^\ast\right)\right) \right\|.$$

Therefore, we have that for all $\vec{\alpha}\in\F_q^k$
\begin{align}
\|\M\left(  \vec{\alpha},g\left(\vec{z}^\ast\right)\right)\|\le1/2.\label{eqfinal}
\end{align}

From \eqref{prop1} we have that if $\vec{z}^\ast\neq \vec{0}$ then  $\|g(\vec{z}^\ast)\|\ge 2/3$.
We think of $g(\vec{z}^\ast)$ as $(\vec{b}_1,\ldots, \vec{b}_{\ell})$, where $\vec{b}_i\in\F_q^k$, for all $i\in[\ell]$. Since $\|g(\vec{z}^\ast)\|\ge 2/3$, we have that $\underset{i\sim [\ell]}{\Pr}[\vec{b}_i=\vec{0}]\le 1/3$. For every $i\in[\ell]$ and a uniformly random $\vec{\alpha}\in\F_q^k$ we have that  $$\Pr_{\vec{\alpha}\sim\F_q^k}[\langle \vec{\alpha},\vec{b}_i\rangle=0]=\begin{cases}
\frac{1}{q}&\text{ if }\vec{b}_i\neq 0\\
0&\text{ otherwise}\end{cases} .$$
 Thus, we have
\begin{align*} \underset{\vec{\alpha}\sim\F_q^k}{\mathbb{E}}\left[\|\M(\vec{\alpha},g(\vec{z}^\ast))\|\right]\ge \frac{q-1}{q}\cdot \frac{2}{3}>\frac{1}{2}. 
\end{align*}

This implies there exists $ \vec{\alpha}\in\F_q^k$ such that $\|\M(\vec{\alpha},g\left(\vec{z}^\ast\right))\|>1/2$, which contradicts \eqref{eqfinal}, and therefore we have $\vec{z}^\ast= \vec{0}$.

\section{Open Problems}\label{sec:open}

The main open problem left behind from this work is to prove the total FPT-inapproximability of the $k$-Clique problem. Apart from this open problem, we would like to highlight the following two open problems too. 

\paragraph{Parameterized Inapproximability Hypothesis (PIH).} The PIH was putforth in \cite{LRSZ20} and asserts that it is W[1]-hard parameterized by $k$, to decide the satisfiability of gap 2-CSP on $k$ variables and alphabet size $n$. It is easy to show that assuming  Gap-ETH, the above gap 2-CSP instances do not admit FPT-approximation algorithms  (for example see \cite{BGKM18}). Previously, many researchers belived that the way to obtain  inapproximability results for the parameterized $k$-Clique problem must be to first resolve PIH. However, Lin \cite{L21} surprisingly found a route  to prove  inapproximability of the $k$-Clique problem while circumventing past PIH.   Nevertheless, since one may see PIH as a parameterized complexity analogue of the PCP theorem (for NP), it remains an outstanding open problem to be settled.  

\paragraph{ETH lower bound for approximating $k$-Clique.}
In \cite{L21} and this paper, we are primarily interested in proving strong hardness of approximation factors for the $k$-Clique problem under the W[1]$\neq$FPT assumption.  However, can we prove tighter running time lower bounds for approximating $k$-Clique problem under stronger assumptions such as ETH? For example, assuming ETH, can we rule out constant factor approximation algorithms for $k$-Clique problem running in $n^{o(k)}$ time? 
Both \cite{L21} and this paper can only prove a time lower bound of $n^{(\log k)^{\Omega(1)}}$ under ETH, for  approximating the $k$-Clique to constant factors. 
\bibliographystyle{alpha}
\bibliography{refs}

 \appendix
 
 \section{Linearity Testing over $\F_q$}\label{sec:appendix}

In this section we prove the following theorem. 

\begin{theorem}[Linearity Testing][Restatement of Theorem~\ref{thm:lintest}]\label{thm:linearitytest}
Let $q$ be a prime number and $d\in \mathbb{N}$.
Let $f:\F_q^d\to\F_q$ be a scalar respecting function. Let $\epsilon, \delta > 0$ be parameters such that 
$\epsilon \gg \delta \gg \frac{1}{q^{1/3}}$. 
If $f$  passes $\mathcal{T}$ with probability $\varepsilon$, then  there exists an integer $r = O(1/\delta^2)$ and 
\emph{linear} functions $c_1,\ldots ,c_r:\F_q^d\to\F_q$, such that the following holds.
$$
\underset{ (\alphab, \betab) \sim S_{f,\mathcal{T}} }{\Pr}\left[\exists\ \text{unique } j\in[r] \text{ such that }f(\alphab)=c_j(\alphab),
 f(\betab)=c_j(\betab) \right]\ge 1- O\left(\frac{\delta}{\epsilon}\right).
$$
\end{theorem}

Before we prove the above theorem, we note some basic facts in Fourier analysis.  
Let $\omega:=e^{2\pi i/q}$ be the complex $q^{\text{th}}$ root of unity and $\Omega :=\{1,\omega,\ldots,\omega^{q-1}\}$.
In this section, we are interested in the complex vector space of all complex valued functions on $\F_q^d$, equipped with the inner product,
$$
\langle g_1,g_2\rangle:=\underset{\alphab\in F_q^d}{\mathbb{E}}[g_1(\alphab)\overline{g_2(\alphab)}].
$$

%A useful fact that comes in handy many times is the following.
\begin{fact}\label{fact}
For  $z\in\Omega$ we have $\sum_{j=0}^{q-1} z^{j}=q$ if $z=1$ and  $\sum_{j=0}^{q-1} z^{j}=0$ if $z\neq 1$.
\end{fact}

Let $g:\F_q^d\to\Omega$ be a function which is scalar respecting (i.e. $g(j \alphab) = g(\alphab)^j$).  We write $$g(\alphab) = \underset{\rhob\in\F_q^d}{\sum}\hat{g}(\rhob)\chi_{\rhob}(\alphab),$$ 
where $\hat{g}:\F_q^d\to \mathbb{C}$ is the Fourier transform of $g$. Here 
$\chi_{\rhob}(\alphab) = \omega^{\rhob \cdot \alphab}$ are the Fourier characters ($\rhob \cdot \alphab$ is the inner product over $\F_q^d$)
 and we have:
$$
\forall \rhob\in\F_q^d,\ \hat{g}(\rhob):= \langle g,\chi_{\rhob}\rangle.
$$
By Parseval, we have $\sum_{\rhob} |\hat{g}(\rhob)|^2 =1$. 
Next we show that the Fourier coefficients are all real numbers for scalar respecting functions $g$. 

\begin{proposition}
For every $\rhob\in\F_q^d$ we have that $\hat{g}(\rhob)$ is a real number.
\end{proposition}\begin{proof} We note that:
$$
\hat{g}(\rhob) %=\langle g,\chi_{\rhob}\rangle
=\underset{\alphab\in\F_q^d}{\mathbb{E}}[g(\alphab)\overline{\chi_{\rhob}(\alphab)}]
= \underset{\alphab\in\F_q^d, j \in \{1,\ldots,q-1\}}{\mathbb{E}}[g(j\alphab)\overline{\chi_{\rhob}(j\alphab)}]
= \underset{\alphab\in\F_q^d }{\mathbb{E}}\left[ \frac{1}{q-1} \sum_{j=1}^{q-1}  (g(\alphab)\overline{\chi_{\rhob}(\alphab)})^j\right].  
$$
We used the fact that $g$ is scalar respecting and that the distribution of
$\alphab$ and $j \alphab$ is the same for $j \in \{1,\ldots,q-1\}$. The proposition follows by noting that 
the inner sum is always a real number, either $q-1$ or $-1$. 
%We used the fact that $g$ is scalar respecting and that the distribution of 
%$\alphab$ and $j \alphab$ is the same for $j \not= 0$. 
%Next we partition the space $\F_q^d\setminus \{\vec{0}\}$ into lines $L_1,\ldots ,L_t$ all passing through the origin. Fix some $i\in[t]$. Then we %can write the points on the line $L_i$ as $\{\gamma\cdot \alphab_0\mid \gamma\in\F_q\setminus \{0\}\}$, for some $\alphab_0\in \F_q^d$. Let $z:=g( %\alphab_0)\overline{\chi_{\rhob}(  \alphab_0)}\in\mathbb{C}$. Note that we have:
%$$
%\underset{\gamma\in\F_q\setminus \{0\}}{\mathbb{E}}[g(\gamma\cdot \alphab_0)\overline{\chi_{\rhob}(\gamma\cdot \alphab_0)}]=\underset{\gamma\in\F_q\setminus \{0\}}{\mathbb{E}}[g( \alphab_0)^\gamma\overline{\chi_{\rhob}( \alphab_0)^{\gamma}}]=\underset{\gamma\in\F_q\setminus \{0\}}{\mathbb{E}}[(g( \alphab_0)\overline{\chi_{\rhob}( \alphab_0)})^{\gamma}]=\underset{\gamma\in\F_q\setminus \{0\}}{\mathbb{E}}[z^{\gamma}].
%$$
%From Fact~\ref{fact} we have that $\underset{\gamma\in\F_q\setminus \{0\}}{\mathbb{E}}[g(\gamma\cdot \alphab_0)\overline{\chi_{\rhob}(\gamma\cdot \alphab_0)}]$ is some real number, and consequently taking union over all lines, we have that $\underset{\alphab\in\F_q^d}{\mathbb{E}}[g(\alphab)\overline{\chi_{\rhob}(\alphab)}]$ is a real number which implies $\hat{g}(\rhob)$ is real.
\end{proof}

We note that the agreement of a function $g$ with the character $\chi_{\rhob}$ is related to the corresponding Fourier coefficient.  
\begin{proposition}\label{prop:agree}
For every $\rhob\in\F_q^d$ we have that $\underset{\alphab\sim\F_q^d}{\Pr}[g(\alphab) = \chi_{\rhob}(\alphab)]=\frac{1}{q} + \frac{q-1}{q} \hat{g}(\rhob)
%\chi_{\rhob}(\alphab)]=\frac{1}{q}\cdot\left(1+\underset{i\in[q-1]}{\sum}(\hat{g}(i\rhob))^i\right)
$. 
\end{proposition}\begin{proof} We note (using that the distribution of $\alphab$ and $j\alphab$ is the same for $j \in \{1,\ldots,q-1\}$):
\begin{align*}
\underset{\alphab\sim\F_q^d}{\Pr}[g(\alphab)=\chi_{\rhob}(\alphab)]&=\underset{\alphab\in\F_q^d}{\mathbb{E}}\left[\frac{1}{q}\cdot \sum_{j=0}^{q-1}(g(\alphab)\overline{\chi_{\rhob}(\alphab)})^j\right]\\
&=\frac{1}{q} + \frac{1}{q}  \underset{\alphab\in\F_q^d}{\mathbb{E}}\left[  \sum_{j=1}^{q-1}  
g(j\alphab) \overline{\chi_{\rhob}(j\alphab)} \right] \\ 
&=\frac{1}{q} + \frac{1}{q}  \underset{\alphab\in\F_q^d}{\mathbb{E}}\left[  \sum_{j=1}^{q-1}
g(\alphab) \overline{\chi_{\rhob}(\alphab)} \right] \\ 
&= \frac{1}{q} + \frac{q-1}{q} \hat{g}(\rhob).
\qedhere\end{align*}
\end{proof}

\begin{lemma}\label{lem:lintestcoeff}
Let $g_1,g_2,g_3:\F_q^d\to\Omega$ be three scalar preserving functions. Let $g_i=\underset{\rhob\in\F_q^d}{\sum}\hat{g}_i(\rhob)\chi_{\rhob}(\alphab)$ be the Fourier representation of $g_i$, for all $i\in\{1,2,3\}$. Then we have:
$$
\Pr_{(\alphab,\betab)\sim \F_q^d\times \F_q^d}[g_1(\alphab)g_2(\betab)=g_3(\alphab+\betab)]=\frac{1}{q}+\frac{q-1}{q}\sum_{\rhob\in\F_q^d} \hat{g}_1(\rhob)\cdot \hat{g}_2(\rhob)\cdot {\hat{g}_3(\rhob)}.
$$
Also, if the probability on the LHS is at least $\epsilon$ and $g_2=g_3$, then for some $\rhob$, $\hat{g}_1(\rhob) \geq \Omega(\epsilon)$. 
\end{lemma}\begin{proof} We note that the desired probability can be expressed as below and follow the standard calculation: 
\begin{align*}
& \underset{\alphab,\betab\sim\F_q^d}{\Pr}\left[ \frac{1}{q} \sum_{j=0}^{q-1} (g_1(\alphab)g_2(\betab)\overline{g_3(\alphab+\betab)})^j  \right] \\ 
= & \frac{1}{q} + \frac{1}{q} \sum_{j=1}^{q-1}  \underset{\alphab\in\F_q^d}{\mathbb{E}}\left[ g_1(j\alphab)g_2(j\betab)\overline{g_3(j\alphab+j\betab)}  \right]\\
= & \frac{1}{q} + \frac{q-1}{q}  \underset{\alphab\in\F_q^d}{\mathbb{E}}\left[ g_1(\alphab)g_2(\betab)\overline{g_3(\alphab+\betab)}  \right]\\
= & \frac{1}{q} + \frac{q-1}{q}  \underset{\alphab\in\F_q^d}{\mathbb{E}}\left[ \sum_{{\rhob}_1, {\rhob}_2, {\rhob}_3 } 
\hat{g}_1({\rhob}_1)\hat{g}_2({\rhob}_2)  \overline{\hat{g}_3({\rhob}_3)}   \chi_{\rhob_1}(\alphab) \chi_{\rhob_2}(\betab) 
\overline{\chi_{\rhob_3}(\alphab+\betab)} 
 \right]\\
 = & \frac{1}{q} + \frac{q-1}{q}  \underset{\alphab\in\F_q^d}{\mathbb{E}}\left[ \sum_{{\rhob}_1, {\rhob}_2, {\rhob}_3 }
\hat{g}_1({\rhob}_1)\hat{g}_2({\rhob}_2)  \hat{g}_3({\rhob}_3)   \chi_{\rhob_1-\rhob_3}(\alphab) \chi_{\rhob_2-\rhob_3}(\betab)
 \right]\\
 = & \frac{1}{q} + \frac{q-1}{q}   \sum_{\rhob }
\hat{g}_1({\rhob})\hat{g}_2({\rhob})  \hat{g}_3({\rhob}).  
\end{align*} 
We used the fact that $\hat{g}_3(\rhob_3)$ are real numbers and that the expectation over $\alphab, \betab$ vanishes unless 
$\rhob_1 = \rhob_2 = \rhob_3$. 
For the last conclusion, we note that $\sum_{\rhob} \hat{g}_2(\rhob)^2 = 1$.  Hence, if the probability is at least $\epsilon \gg \frac{1}{q}$,
then for some $\rho$, we have $\hat{g}_1(\rhob) \geq \Omega(\epsilon)$.  \qedhere 
\end{proof}

\begin{proof}[Proof of Theorem~\ref{thm:linearitytest}] Let $f, \epsilon, \delta$ be as in the statement of the theorem. 
Given $f$, we define the function $g:\F_q^d\to\Omega$ as follows:
$$
\forall \alphab\in\F_q^d,\ g(\alphab):=\omega^{f(\alphab)}.$$
Under this transformation, linear functions on $\F_q^d$ can equivalently be viewed as Fourier characters, and we use the two views interchangeably. 
Clearly, for every $\alphab,\betab\in\F_q^d$, 
$
f(\alphab)+f(\betab)=f(\alphab+\betab)\Longleftrightarrow g(\alphab)g(\betab)=g(\alphab+\betab)
$. Since $f$ passes the test ${\mathcal T}$ with probability $\epsilon$, we have
  $$ \Pr_{(\alphab,\betab)\sim \F_q^d\times \F_q^d}[g(\alphab)g(\betab)=g(\alphab+\betab)] \geq \epsilon. $$ 
From Lemma \ref{lem:lintestcoeff} (applied with $g_1=g_2=g_3=g$), we conclude that for some $\rhob$, we have $\hat{g}(\rhob) \geq \Omega(\epsilon)$.
%Next, we define functions $\varphi:\F_q^d\times \F_q^d\to \Omega$ and $\psi:\F_q^d\times \F_q^d\to \{0,1\}$ as follows. For all $\alphab,\betab\in\F_q^d$ we have
%$$
%\varphi(\alphab,\betab):=g(\alphab)\cdot  g(\betab) \cdot\overline{g(\alphab+\betab)}\text{ and }\psi(\alphab,\betab):=\frac{1+\varphi(\alphab,\betab)+(\varphi(\alphab,\betab))^2+\cdots +(\varphi(\alphab,\betab))^{q-1}}{q}.
%$$
%It is easy to see that $\psi$ always maps to 0 or 1 from Fact~\ref{fact}. Thus we have that
%$$
%\mathbb{E}_{(\alphab,\betab)\in \F_q^d\times \F_q^d}[\psi(\alphab,\betab)]=\Pr_{(\alphab,\betab)\in \F_q^d\times \F_q^d}[f(\alphab)+f(\betab)=f(\alphab+\betab)]=\varepsilon.
%$$
%
%
%
%We now write $g(\alphab)$ as its fourier transform $\underset{\rhob\in\F_q^d}{\sum}\hat{g}(\rhob)\chi_{\rhob}(\alphab)$. From Lemma~\ref{lem:lintestcoeff} we have:
%\begin{align*}
%\varepsilon=\frac{1}{q}+\frac{q-1}{q}\sum_{\rhob\in\F_q^d}(\hat{g}(\rhob))^3
%\end{align*}
%
%If $\varepsilon=\omega(1/q)$ then we have
%
%
%
%There exists $\rhob\in\F_q^d$ such that $|\hat{g}(\rhob)|\ge \varepsilon'$. 
We now list the set of all characters (i.e. linear functions) for which $g$ has a large corresponding Fourier coefficient (the constant 
in $\Omega(\delta)$ is chosen small enough so that the subsequent argument works):  
% $g$ on at least $1/\varepsilon^2$ fraction of points:
$$ {\sf List} := \{\chi_{\rhob} \in \F_q^d  ~| ~ |\hat{g}(\rhob)| \geq \Omega(\delta)  \}.  
%\underset{\alphab\sim\F_q^d}{\Pr}[
%g(\alphab)=\chi_{\rhob}(\alphab)]\ge \varepsilon^2\right\}.
$$
Since the squared Fourier coefficients sum upto $1$, the list is bounded, i.e. $|{\sf List}| \leq O(1/\delta^2)$. Now suppose on the 
contrary that, with probability $\delta$, $g$ passes the linearity test but the value $g(\alphab)$ is inconsistent with all linear
functions $\chi_{\rhob}$ in the list, i.e. 
%
%
%However, from Proposition~\ref{prop:agree}, we have that for all $\chi_{\rhob}\in R^\ast$ we have $\hat{g}(\rhob)>\varepsilon^2-\frac{1}{q}$. Applying Parseval's inequality yields us that $|R^\ast|\le 4/\varepsilon^4$.
%
%Next, suppose that
\begin{equation} \label{eqn:inconsistent}
\Pr_{\alphab,\betab } \left[ g(\alphab) g(\betab)= g(\alphab+\betab) \ \mbox{and} \ g(\alphab) \not\in \{ \chi_{\rhob}(\alphab) ~| ~ 
\rhob \in {\sf List} \}    
%\nexists\ \chi_{\rhob}\in R^\ast\text{ such that }\chi_{\rhob}(\alphab)=g(\alphab)
\right] \geq \delta.
\end{equation} 

Let $X$ be the subset of $\F_q^d$ on which $g(\alphab) \not\in \{ \chi_{\rhob}(\alphab) ~| ~
\rhob \in {\sf List} \}$. Consider a new function $\tilde{g}:\F_q^d\to\Omega$ defined\footnote{For the sake of conciseness, we skip here the additional care that needs to be taken while defining  $\tilde{g}$ to ensure that it is scalar preserving. } as follows:
$$
\tilde{g}(\alphab)=\begin{cases}
g(\alphab) & if \ \alphab\in X\\
\text{random} & \text{otherwise}
\end{cases}
$$

Thus from Equation \eqref{eqn:inconsistent}, 
$$\Pr_{(\alphab,\betab)\sim \F_q^d\times\F_q^d}[\tilde{g}(\alphab)g(\betab)=g(\alphab+\betab)] \geq \delta.$$

From Lemma \ref{lem:lintestcoeff} (applied with $g_1 = \tilde{g}, g_2=g_3=g$), 
there exists $\rhob^*$ such that $\hat{\tilde{g}}(\rhob^*) \geq \Omega(\delta)$. From 
Proposition \ref{prop:agree}, we have that $\tilde{g}$ has agreement $\Omega(\delta)$ with $\chi_{\rhob^*}$. Since $\tilde{g}$ is defined randomly
outside $X$, this agreement must essentially be on $X$. Further, since  $\tilde{g}$ coincides with $g$ on $X$, $g$ has agreement  $\Omega(\delta)$ 
with $\chi_{\rhob^*}$, and 
hence  by Proposition \ref{prop:agree} again, 
 $\hat{g}(\rhob^*) \geq \Omega(\delta)$ and thus $\rhob^* \in {\sf List}$. This is a contradiction since
$g$ is not supposed to agree with any linear function in the list on $X$.

Thus, we conclude that  
$$\Pr_{\alphab,\betab } \left[ g(\alphab) g(\betab)= g(\alphab+\betab) \ \mbox{and} \ g(\alphab) \not\in \{ \chi_{\rhob}(\alphab) ~| ~
\rhob \in {\sf List} \}
%\nexists\ \chi_{\rhob}\in R^\ast\text{ such that }\chi_{\rhob}(\alphab)=g(\alphab)
\right] \leq \delta.$$
The same statement applies with $\betab$ in place of $\alphab$. Moreover, the fraction of points in $\F_q^d$ on which two distinct 
$\chi_{\rhob_1}, \chi_{\rhob_2} \in {\sf List}$ agree is at most $O(|{\sf List}|^2/q)$, i.e. $O(1/(\delta^2q))$, i.e.
$O(\delta)$ since $\delta \gg \frac{1}{q^{1/3}}$. Finally, 
we condition on the $\epsilon$-probability event that $g(\alpha)+g(\beta)=g(\alpha+\beta)$ and 
conclude that 
$$\Pr_{(\alphab,\betab) \sim S_{f, {\mathcal T} }}  
\left[  g(\alphab) = \chi_{\rhob}(\alphab), 
g(\betab) = \chi_{\rhob}(\betab) \ \mbox{for some unique} \ \rhob \in {\sf List}  
%\nexists\ \chi_{\rhob}\in R^\ast\text{ such that }\chi_{\rhob}(\alphab)=g(\alphab)
\right] \geq 1- O\left(\frac{\delta}{\epsilon}\right).$$
\end{proof}

\end{document}